\documentclass{article}
\usepackage[left=1in, right=1in, top=1in, bottom=1in]{geometry}
\usepackage{lineno}
\usepackage{enumitem}
\usepackage{authblk}
\usepackage{amsmath, amsthm, amssymb}
\usepackage{xspace}
\usepackage{comment}
\usepackage{hyperref}
\usepackage{xcolor}

\newcommand{\algprobm}[1]{\textsc{#1}\xspace}

\theoremstyle{plain}
\newtheorem{theorem}{Theorem}[section]
\newtheorem{proposition}[theorem]{Proposition}
\newtheorem{corollary}[theorem]{Corollary}
\newtheorem{lemma}[theorem]{Lemma}
\newtheorem{observation}[theorem]{Observation}

\theoremstyle{definition}
\newtheorem{definition}[theorem]{Definition}
\newtheorem{remark}[theorem]{Remark}
\newtheorem{question}[theorem]{Question}

\newcommand{\Lem}[1]{Lemma~\ref{#1}\xspace}
\newcommand{\Cor}[1]{Corollary~\ref{#1}\xspace}

\newcommand{\Thm}[1]{Theorem~\ref{#1}\xspace}

\DeclareMathOperator{\Aut}{Aut}

\DeclareMathOperator{\poly}{poly}

\title{On the Complexity of Identifying Strongly Regular Graphs}

\author[1]{Michael Levet
\footnote{
This work was partially supported by J. A. Grochow's startup funds and NSF award CISE-2047756. I wish to thank C.J. Colbourn for helping me to find a copy of \cite{colborun_1979}. I also wish to thank J.N. Cooper and J.A. Grochow for helpful discussions. 
}
}
\affil[1]{
Department of Computer Science- University of Colorado Boulder
}

\begin{document}
\maketitle

\begin{abstract}
In this paper, we show that \algprobm{Graph Isomorphism} (\algprobm{GI}) is not $\textsf{AC}^{0}$-reducible to several problems, including the \algprobm{Latin Square Isotopy} problem, isomorphism testing of several families of Steiner designs, and isomorphism testing of conference graphs. As a corollary, we obtain that $\algprobm{GI}$ is not $\textsf{AC}^{0}$-reducible to isomorphism testing of Latin square graphs and strongly regular graphs arising from special cases of Steiner $2$-designs. We accomplish this by showing that the generator-enumeration technique for each of these problems can be implemented in $\beta_{2}\textsf{FOLL}$, which cannot compute $\algprobm{Parity}$ (Chattopadhyay, Tor\'an, \& Wagner, \textit{ACM Trans. Comp. Theory}, 2013).  
\end{abstract}


\thispagestyle{empty}

\newpage

\setcounter{page}{1}

\section{Introduction}
\label{sec:introduction}

In one of the original papers on \algprobm{Graph Isomorphism} (\algprobm{GI}), by Corneil \& Gotlieb, the authors observed that in practice, strongly regular graphs serve as difficult instances \cite{CorneilGotlieb}. A common approach for producing such instances is to construct strongly regular graphs from combinatorial objects, such as Latin squares, nets (partial geometries), and combinatorial designs \cite{CorneilGotlieb, SchmidtDruffel}. While strongly-regular graphs have been perceived as hard instances, there is little evidence to suggest that they are $\textsf{GI}$-complete.

There is a beautiful classification of strongly-regular graphs (of valency $\rho < n/2$, which is essentially without loss of generality by taking complements) due to  Neumaier \cite{Neumaier}. In using this classification in isomorphism testing, Spielman \cite{SpielmanSRGs} and Babai--Wilmes \cite{BabaiWilmes} organize it as follows:
\begin{enumerate}[label=(\alph*)]
\item Latin square graphs,
\item Line graphs of Steiner $2$-designs satisfying $\rho < f(n)$, for a certain function $f(n) \sim n^{3/4}$,
\item Strongly-regular graphs of degree $\rho = (n-1)/2$ (a.k.a. conference graphs),
\item Graphs satisfying a certain eigenvalue inequality referred to as \textit{Neumaier's claw bound}.
\end{enumerate} 

In this paper, we show (in particular) that (a) Latin square graphs, (b) line graphs arising from special cases of Steiner $2$-designs, and (c) conference graphs are not $\textsf{GI}$-hard under $\textsf{AC}^{0}$-computable many-one reductions. As with \cite{ChattopadhyayToranWagner}, we show this by improving the parallel complexity of isomorphism testing in these classes of graphs to $\beta_{2}\textsf{FOLL}$, which does not compute $\algprobm{Parity}$ (in the case of conference graphs, we achieve a stronger bound of $\beta_{2}\textsf{AC}^{0}$). As $\algprobm{Parity}$ is $\textsf{AC}^{0}$-reducible to $\algprobm{GI}$ \cite{Toran}, this rules out $\textsf{AC}^{0}$-reductions (and more strongly, for any $i, c \geq 0$, we rule out $\beta_{i}\textsf{FO}((\log \log n)^{c})$-reductions).

Prior to our work, there were few pieces of complexity-theoretic evidence to suggest that strongly-regular graphs are not $\textsf{GI}$-complete. Babai showed that there is no functorial reduction from \algprobm{GI} to the isomorphism testing of strongly-regular graphs \cite{BabaiFunctorialSRGs}. We note that almost all known reductions between isomorphism problems are functorial (c.f. \cite{BabaiFunctorialSRGs}). An example where this is not the case is the reduction from $1$-$\algprobm{Block Conjugacy}$ of shifts of finite type to $k$-\algprobm{Block Conjugacy} \cite[Theorem~18]{SchrockFrongillo}. In the case of \algprobm{Quasigroup Isomorphism}, Chattopadhyay, Tor\'an, and Wagner improved the upper bound to $\beta_{2}\textsf{L} \cap \beta_{2}\textsf{FOLL}$. In particular, they showed that for any $i, c \geq 0$, \algprobm{GI} is not $\beta_{i}\textsf{FO}((\log \log n)^{c})$-reducible to \algprobm{Quasigroup Isomorphism} \cite{ChattopadhyayToranWagner}.

Let us briefly discuss why problems such as \algprobm{Quasigroup Isomorphism}, \algprobm{Latin Square Isotopy}, and isomorphism testing of Steiner $2$-designs are important from a complexity-theoretic perspective. Several families of strongly regular graphs (e.g., Latin square graphs, and the block-intersection graphs of Steiner $2$-designs) arise naturally from the isomorphism problems from the underlying combinatorial objects. Precisely, many of these problems, including \algprobm{Quasigroup Isomorphism}, \algprobm{Latin Square Isotopy}, and isomorphism testing of Steiner $2$-designs, are polynomial-time (and in fact, $\textsf{AC}^{0}$) reducible to \algprobm{GI}. Polynomial-time solutions for these problems have been elusive. In particular, each of these problems are candidate $\textsf{NP}$-intermediate problems; that is, problems that belong to $\textsf{NP}$, but are neither in $\textsf{P}$ nor $\textsf{NP}$-complete \cite{Ladner}. 

There is considerable evidence suggesting that $\algprobm{GI}$ is not $\textsf{NP}$-complete \cite{Schoning, BuhrmanHomer, ETH, BabaiGraphIso, GILowPP, ArvindKurur, MATHON1979131}. As \algprobm{Quasigroup Isomorphism}, \algprobm{Latin Square Isotopy}, and isomorphism testing of Steiner $2$-designs reduce to $\algprobm{GI}$, this evidence also suggests that none of these problems is $\textsf{NP}$-complete. Furthermore, Chattopadhyay, Tor\'an, \& Wagner \cite{ChattopadhyayToranWagner} showed that \algprobm{Quasigroup Isomorphism} is not $\textsf{NP}$-complete under $\textsf{AC}^{0}$-reductions, and our result shows the same for \algprobm{Latin Square Isotopy}. In light of Babai's $n^{\Theta(\log^{2}(n))}$-time algorithm \cite{BabaiGraphIso}, these problems serve as key barries to placing \algprobm{GI} into \textsf{P}.  \\

\noindent \textbf{Main Results.} In this paper, we show that several families of isomorphism problems characterized by the generator enumeration technique are not $\textsf{GI}$-complete under $\beta_{i}\textsf{FO}((\log \log n)^{c})$-reductions, for any choice of $i, c \geq 0$. In particular, this rules out $\textsf{AC}^{0}$-reductions from $\algprobm{GI}$ to these problems. We also show that these results extend to the corresponding families of strongly regular graphs, providing complexity-theoretic evidence that isomorphism testing of strongly regular graphs may not be $\textsf{GI}$-complete.

The first problem we consider is the \algprobm{Latin Square Isotopy} problem, which takes as input two quasigroups $L_{1}, L_{2}$ given by their multiplication tables and asks if there exist bijections $\alpha, \beta, \gamma : L_{1} \to L_{2}$ such that $\alpha(x)\beta(y) = \gamma(xy)$ for all $x, y \in L_{1}$.

\begin{theorem} \label{thm:LatinSquareIsotopy}
$\algprobm{Latin Square Isotopy} \in \beta_{2}\textsf{L} \cap \beta_{2}\textsf{FOLL}$.
\end{theorem}

\begin{remark}
This improves the previous bound of $\beta_{2}\textsf{NC}^{2}$ due to Wolf \cite{Wolf}.
\end{remark}

To prove Theorem \ref{thm:LatinSquareIsotopy}, we leverage cube generating sequences, in a similar manner as Chattopadhyay, Tor\'an, \& Wagner \cite[Theorem~3.4]{ChattopadhyayToranWagner}. After non-deterministically guessing cube generating sequences, we can check in $\textsf{L} \cap \textsf{FOLL}$ whether the induced map extends to an isotopy of the Latin squares.

Now for any $i, c \geq 0$, we have that $\beta_{i}\textsf{FO}( (\log \log n)^{c})$ cannot compute \algprobm{Parity} \cite{ChattopadhyayToranWagner}. Thus, we obtain the following corollary.

\begin{corollary}
For any $i, c \geq 0$, \algprobm{GI} is not $\beta_{i}\textsf{FO}((\log \log n)^{c})$-reducible to \algprobm{Latin Square Isotopy}.
\end{corollary}

Latin squares yield a corresponding family of strongly regular graphs, known as \textit{Latin square graphs}, where two Latin square graphs $G_{1}$ and $G_{2}$ are isomorphic if and only if their corresponding Latin squares $L(G_{1})$ and $L(G_{2})$ are main class isotopic \cite[Lemma~3]{MillerTarjan} (that is, the Latin squares can be put into \textit{compatible} normal forms that correspond to isomorphic quasigroups). Miller previously showed that it is possible to recover the corresponding Latin square from a Latin square graph in polynomial-time \cite{MillerTarjan}, and Wolf strengthened the analysis to show that this reduction is in fact $\textsf{NC}^{1}$-computable \cite{Wolf}. A closer analysis yields that this reduction is actually $\textsf{AC}^{0}$-computable (see Lemma~\ref{RecoverLSGAC0}). Furthermore, we can place a Latin square into its main isotopy class in $\textsf{AC}^{0}$ (see Remark~\ref{rmk:MainClass}). Thus, we obtain the following corollary.

\begin{corollary} \label{cor:LatinSquareGI}
Deciding whether two Latin square graphs are isomorphic is in $\beta_{2}\textsf{L} \cap \beta_{2}\textsf{FOLL}$. Consequently, for any $i, c \geq 0$, \algprobm{GI} is not $\beta_{i}\textsf{FO}((\log \log n)^{c})$-reducible to isomorphism testing of Latin square graphs.
\end{corollary}

\begin{remark}
It is also possible to construct a Latin square graph from a Latin square using an $\textsf{AC}^{0}$-circuit. Thus, \algprobm{Latin Square Isotopy} and isomorphism-testing of Latin square graphs are equivalent under $\textsf{AC}^{0}$-reductions.
\end{remark}

\begin{remark} \label{rmk:MainClass}
We sketch an alternative proof of \Thm{thm:LatinSquareIsotopy} and \Cor{cor:LatinSquareGI}. In Miller's \cite{MillerTarjan} first proof that \algprobm{Latin Square Isotopy} is solvable in time $n^{\log(n) + O(1)}$, Miller takes one Latin square $L'$ and places it in a normal form. Miller shows that this step is polynomial-time computable. A closer analysis shows that it is in fact $\textsf{AC}^{0}$-computable: we may label the first row and first column as $1, 2, \ldots, n$. We then in $\textsf{AC}^{0}$ fill in the remaining cells of the normal form. For the second Latin square $L$, Miller places $L$ in $n^{2}$ possible normal forms by guessing the element to label as $1$. We may consider all such guesses in parallel, and so this step is also $\textsf{AC}^{0}$-computable. Miller then checks whether the normal form of $L'$ and the normal form of $L$ are isomorphic as quasigroups. This step is $\beta_{2}\textsf{L} \cap \beta_{2}\textsf{FOLL}$-computable \cite{ChattopadhyayToranWagner}. Thus, we have an $\textsf{AC}^{0}$-computable disjunctive truth-table reduction (see Section~\ref{sec:Complexity}) from \algprobm{Latin Square Isotopy} to \algprobm{Quasigroup Isomorphism}, which places \algprobm{Latin Square Isotopy} into $\beta_{2}\textsf{L} \cap \beta_{2}\textsf{FOLL}$. We note that Wolf's \cite{Wolf} proof that $\algprobm{Latin Square Isotopy} \in \beta_{2}\textsf{NC}^{2}$ followed a similar strategy as Miller's \cite{MillerTarjan} proof that \algprobm{Quasigroup Isomorphism} can be solved in time $n^{\log(n) + O(1)}$.

As we can recover a Latin square from a Latin square graph in $\textsf{AC}^{0}$ (see Lemma~\ref{RecoverLSGAC0}), we also have an $\textsf{AC}^{0}$-computable disjunctive truth table reduction from isomorphism testing of Latin square graphs to \algprobm{Quasigroup Isomorphism}. So isomorphism testing of Latin square graphs belongs to $\beta_{2}\textsf{L} \cap \beta_{2}\textsf{FOLL}$.

Instead, the proof we give of \Thm{thm:LatinSquareIsotopy} (see Section 3) exhibits how the technique of cube generating sequences \cite{ChattopadhyayToranWagner} can be fruitfully applied directly to isotopy (rather than isomorphism) of algebraic structures, which we do by combining cube generating sequences with Miller's \cite{MillerTarjan} second proof that \algprobm{Latin Square Isotopy} can be solved in time $n^{\log(n) + O(1)}$.
\end{remark}

We now turn our attention to isomorphism testing of Steiner designs.

\begin{theorem} \label{ThmMainGenEnum}
For any $i, c \geq 0$, \algprobm{GI} is not $\beta_{i}\textsf{FO}((\log \log n)^{c})$-reducible to isomorphism testing of Steiner $(t, t+1)$-designs. 
\end{theorem}

We prove Theorem \ref{ThmMainGenEnum} in two steps. First, we establish the case of $t = 2$, which is the case of Steiner triple systems. It is well-known that Steiner triple systems correspond to quasigroups, and that this reduction is polynomial-time computable \cite{MillerTarjan}. A careful analysis shows that this reduction is in fact $\textsf{AC}^{0}$-computable. Furthermore, we observe that the reduction in \cite[Theorem~33]{BabaiWilmes} from isomorphism testing of Steiner $t$-designs to isomorphism testing of Steiner $2$-designs is $\beta_{2}\textsf{AC}^{0}$-computable. As a corollary, isomorphism testing of Steiner $(t, t+1)$-designs is $\beta_{2}\textsf{AC}^{0}$-reducible to isomorphism testing of Steiner triple systems, and hence \algprobm{Quasigroup Isomorphism}.

Steiner designs yield a family of graphs known as \textit{block intersection graphs}. In the case of Steiner $2$-designs, these graphs are strongly regular. When the block size is bounded, we observe that we may recover a Steiner $2$-design from its block-incidence graph in $\textsf{AC}^{0}$. If the block size is not bounded, this reduction is $\textsf{TC}^{0}$-computable. In the case of Steiner triple systems, this yields the following corollary.

\begin{corollary}
Let $G_{1}, G_{2}$ be block-incidence graphs arising from Steiner triple systems. We can decide whether $G_{1} \cong G_{2}$ in $\beta_{2}\textsf{L} \cap \beta_{2}\textsf{FOLL}$. Consequently, for any $i, c \geq 0$, \algprobm{GI} is not $\beta_{i}\textsf{FO}((\log \log n)^{c})$-reducible to isomorphism testing of block-incidence graphs arising from Steiner triple systems.
\end{corollary}

We finally consider the case of conference graphs.

\begin{theorem} \label{thm:MainConference}
Isomorphism testing between a conference graph $G$ and an arbitrary graph $H$ belongs to $\beta_{2}\textsf{AC}^{0}$. Consequently, for any $i, c \geq 0$, there is no $\beta_{i}\textsf{FO}((\log \log n)^{c})$-reduction from \algprobm{GI} to isomorphism testing of conference graphs.
\end{theorem}

To prove \Thm{thm:MainConference}, we use a result of Babai \cite{BabaiCanonicalLabeling1}, who showed that conference graphs admit so called \textit{distinguishing sets} (see Section~\ref{sec:Conference}) of size $O(\log n)$. Distinguishing sets are quite powerful: after indvidiualizing the vertices in a distinguishing set, only two rounds of the count-free Color Refinement algorithm are needed to assign each vertex in the graph a unique label. The $\beta_{2}\textsf{AC}^{0}$ algorithm thus works as follows: we first guess such a set $S$ using $O(\log^{2} n)$ non-deterministic bits, individualize the vertices in $S$, and then run two rounds of the count-free Color Refinement algorithm. 

\begin{remark}
\Thm{thm:MainConference} provides evidence that isomorphism testing of conference graphs might in fact be easier than \algprobm{Group Isomorphism} (and hence, isomorphism testing of Latin square graphs and Steiner designs). The best known complexity-theoretic upper bound for \algprobm{Group Isomorphism} is $\beta_{2}\textsf{L} \cap \beta_{2}\textsf{FOLL} \cap \beta_{2}\textsf{SC}^{2}$ (the $\beta_{2}\textsf{L} \cap \beta_{2}\textsf{FOLL}$ is due to Chattopadhyay, Tor\'an, \& Wagner \cite{ChattopadhyayToranWagner}; Tang \cite{TangThesis} established the $\beta_{2}\textsf{SC}^{2}$ bound and provided an alternative proof of the $\beta_{2}\textsf{L}$ bound). Even after fixing a cube generating sequence, it is not clear that an $\textsf{AC}$ circuit of depth $o(\log \log n)$ and polynomial size could uniquely identify each group element.

These differences are quite opaque when examined from the algorithmic perspective. Both conference graphs and groups admit $n^{\Theta(\log n)}$-time isomorphism tests that have seen little improvement in the worst case in over 40 years. For conference graphs, this bound is achieved using Babai's \cite{BabaiCanonicalLabeling1} result showing the existence of size $O(\log n)$ distinguishing sets. In the setting of groups, Tarjan (c.f., \cite{MillerTarjan}) leveraged the fact that every group admits a generating set of size $\leq \lceil \log_{p} n \rceil$ (where $p$ is the smallest prime dividing $n$) to obtain an $n^{\log_{p}(n) + O(1)}$ isomorphism test via the generator enumeration strategy. Rosenbaum \cite{Rosenbaum2013BidirectionalCD} (see \cite[Sec. 2.2]{GR16}) improved this to $n^{(1/4)\log_p(n) + O(1)}$. And even the impressive body of work on practical algorithms for this problem, led by Eick, Holt, Leedham-Green and O'Brien (e.\,g., \cite{BEO02, ELGO02, BE99, CH03}) still results in an $n^{\Theta(\log n)}$-time\footnote{Where we stress $n = |G|$. In the setting of practical algorithms, the groups are often given by generating sequences, and so polynomial time in the size of the succinct input is $\poly \log |G|$.} algorithm in the general case (see \cite[Page 2]{WilsonSubgroupProfiles}).
\end{remark}

\noindent \textbf{Further Related Work.} There has been significant work on isomorphism testing of strongly regular graphs. In 1980, Babai used a simple combinatorial approach to test strongly regular graphs on $n$ vertices and degree $\rho < n/2$ in time $\text{exp}(O( (n/\rho) \log^{2} n))$ \cite{BabaiCanonicalLabeling1, BabaiCanonicalLabeling2}. We note that the complement of a strongly-regular graph is strongly regular, hence the assumptions on the degree. As $\rho \geq \sqrt{n-1}$, this gives an algorithm in moderately exponential time $\text{exp}(O(\sqrt{n} \log^{2} n))$ for all $\rho$. Spielman \cite{SpielmanSRGs} subsequently improved this $\text{exp}(O( (n/\rho) \log^{2} n))$ bound to $\exp(O(n^{1/3} \log^{2} n))$. This improved upon what was at the time, the best known bound of $\exp(O(\sqrt{n \log n}))$ for isomorphism testing of general graphs \cite{BabaiLuksCanonicalLabeling, BabaiKantorLuksCFSG, Zemlyachenko1985} based on Luks' group theoretic method \cite{LuksBoundedValence}. 

We note that Babai's \cite{BabaiCanonicalLabeling1} work already handled the case of conference graphs in time $n^{O(\log n)}$. Furthermore, the works of Babai and Spielman sufficed to handle isomorphism testing of strongly regular graphs that satisfy Neumaier's claw bound. Namely, Spielman  \cite{SpielmanSRGs} handled the case of $\rho \in o(n^{2/3})$ in time $\text{exp}(O(n^{1/4} \log^{2} n))$, and Babai \cite{BabaiCanonicalLabeling1} handled the case of $\rho \in \Omega(n^{2/3})$ in time $\text{exp}(O(n^{1/3} \log^{2} n))$.

In the case of Latin square graphs, Miller \cite{MillerTarjan} resolved this in time $n^{O(\log n)}$. Wolf improved the complexity theoretic bound to $\beta_{2}\textsf{NC}^{2}$ (uniform $\textsf{NC}^{2}$-circuits that accept $O(\log^{2} n)$ existentially quantified non-deterministic bits). 

The isomorphism problem for Steiner $2$-designs also dates back to Miller \cite{MillerTarjan}, who considered the special cases of Steiner triple systems (Steiner $(2,3,n)$-designs), projective planes (Steiner $(2, q+1, q^{2} + q + 1)$-designs), and affine planes (Steiner $(2, q, q^{2})$-designs). In the case of Steiner triple systems, Miller leveraged a standard reduction to \algprobm{Quasigroup Isomorphism} to obtain an $n^{O(\log n)}$ upper bound. M. Colbourn \cite{colborun_1979} subsequently extended this result to obtain an $v^{O(\log v)}$ canonization procedure for Steiner $(t, t+1)$-designs.

For the cases of projective and affine planes, Miller \cite{MillerTarjan} showed that isomorphism testing can be done in time $n^{O(\log \log n)}$. In particular, Miller showed that the corresponding structures may be recovered from the block-intersection graphs in polynomial-time, providing the block-intersection graphs are of sufficiently high degree \cite{MillerTarjan}. Thus, isomorphism testing of block-intersection graphs arising from Steiner triple systems can be done in time $n^{O(\log n)}$, and the block-intersection graphs arising from projective and affine planes can be identified in time $n^{O(\log \log n)}$. 

In 1983, Babai \& Luks showed that Steiner $2$-designs with blocks of bounded size admit an $n^{O(\log n)}$ canonization procedure (and hence, isomorphism test). In paritcular, isomorphism testing of the corresponding block-intersection graphs is solvable in time $n^{O(\log n)}$ \cite{BabaiLuksCanonicalLabeling}. Here, Babai \& Luks utilized Luks' group theoretic method \cite{LuksBoundedValence}. Huber later extended this result to obtain an $n^{O(\log n)}$-runtime isomorphism test for arbitrary $t$-designs, where both $t$ and the block size are bounded, as well as the corresponding block-intersection graphs \cite{Huber}. 

Spielman solved the general case of strongly-regular graphs of degree $f(n) \sim n^{3/4}$ arising from Steiner $2$-designs in time $\text{exp}(O(n^{1/4} \log^{2} n))$ \cite{SpielmanSRGs}. Independently, Babai \& Wilmes \cite{BabaiWilmes} and Chen, Sun, \& Teng \cite{ChenSunTeng} exhibited an $n^{O(\log n)}$-runtime isomorphism test for both Steiner $2$-designs and the corresponding block-intersection graphs. Babai \& Wilmes extended their result to Steiner $t$-designs. Here, each of these papers utilized the individualize and refine technique \cite{SpielmanSRGs, BabaiWilmes, ChenSunTeng}.

The isomorphism problem for combinatorial designs is another candidate $\textsf{NP}$-intermediate problem. In the general case, testing whether two combinatorial designs are isomorphic is \textsf{GI}-complete under polynomial-time Turing reductions \cite{ColbournColbournGI}. Cyclic Steiner $2$-designs of prime order are known to admit a polynomial-time isomorphism test \cite{ColbournMathon}. To the best of our knowledge, the complexity of deciding whether two cyclic Steiner $2$-designs are isomorphic remains open for arbitrary orders.

\section{Preliminaries} \label{sec:Preliminaries}

\subsection{Algebra and Combinatorics}

\textbf{Graph Theory}. A \textit{strongly regular graph} with parameters $(n, k, \lambda, \mu)$ is a simple, undirected $k$-regular, $n$-vertex graph $G(V, E)$ where any two adjacent vertices share $\lambda$ neighbors, and any two non-adjacent vertices share $\mu$ neighbors. The complement of a strongly regular graph is also strongly regular, with parameters $(n, n-k-1, n-2-2k+\mu, n-2k+\lambda)$. \\

\noindent \textbf{Quasigroups}. A \textit{quasigroup} consists of a set $G$ and a binary operation $\star : G \times G \to G$ satisfying the following. For every $a, b \in G$, there exist unique $x, y$ such that $a \star x = b$ and $y \star a = b$. When the multiplication operation is understood, we simply write $ax$ for $a \star x$.

Unless otherwise stated, all quasigroups are assumed to be finite and represented using their Cayley (multiplication) tables. For a quasigroup of order $n$, thet Cayley table has $n^{2}$ entries, each represented by a string of size $\lceil \log_{2}(n) \rceil$. 

A \textit{Latin square} of order $n$ is an $n \times n$ matrix $L$ where each cell $L_{ij} \in [n]$, and each element of $[n]$ appears exactly once in a given row or a given column. Latin squares are precisely the Cayley tables corresponding to quasigroups. We will abuse notation by referring to a quasigroup and its multiplication table interchangeably. An \textit{isotopy} of Latin squares $L_{1}$ and $L_{2}$ is an ordered triple $(\alpha, \beta, \gamma)$, where $\alpha, \beta, \gamma : L_{1} \to L_{2}$ are bijections satisfying the following: whenever $ab = c$ in $L_{1}$, we have that $\alpha(a)\beta(b) = \gamma(c)$ in $L_{2}$. Alternatively, we may view $\alpha$ as a permutation of the rows of $L_{1}$, $\beta$ as a permutation of the rows of $L_{2}$, and $\gamma$ as a permutation of the values in the table. Here, $L_{1}$ and $L_{2}$ are isotopic precisely if (i) the $(i,j)$ entry of $L_{1}$ is the $(\alpha(i), \beta(j))$ entry of $L_{2}$, and (ii) $x$ is the $(i,j)$ entry of $L_{1}$ if and only if $\gamma(x)$ is the $(\alpha(i), \beta(j))$ entry of $L_{2}$.

As quasigroups are non-associative, the parenthesization of a given expression may impact the resulting value. We restrict attention to balanced parenthesizations, which ensure that words of the form $g_{0}g_{1} \cdots g_{k}$ are evaluated using a balanced binary tree with $k+1$ leves and therefore depth $O(\log \log n)$.  

For a sequence $S := (s_{0}, s_{1}, \ldots, s_{k})$ from a quasigroup, define:
\[
\text{Cube}(S) = \{ s_{0}s_{1}^{e_{1}} \cdots s_{k}^{e_{k}} : e_{1}, \ldots, e_{k} \in \{0,1\} \}.
\]

We say that $S$ is a \textit{cube generating sequence} if each element $g$ in the quasigroup can be written as $g = s_{0}s_{1}^{e_{1}} \cdots s_{k}^{e_{k}}$, for $e_{1}, \ldots, e_{k} \in \{0,1\}$. Every quasigroup is known to admit a cube generating sequence of size $O(\log n)$ \cite{ChattopadhyayToranWagner}.

For a given Latin square $L$ of order $n$, we associate a \textit{Latin square graph} $G(L)$ that has $n^{2}$ vertices; one for each triple $(a, b, c)$ that satisfies $ab = c$. Two vertices $(a, b, c)$ and $(x, y, z)$ are adjacent in $G(L)$ precisely if $a = x$, $b = y$, or $c = z$. Miller showed that two Latin square graphs $G_{1}$ and $G_{2}$ are isomorphic if and only if the corresponding Latin squares, $L_{1}$ and $L_{2}$, are \textit{main class isotopic}; that is, if $L_{1}$ and $L_{2}$ can be placed into compatible normal forms that correspond to isomorphic quasigroups \cite{MillerTarjan}. A Latin square graph on $n^{2}$ vertices is a strongly regular graph with parameters $(n^{2}, 3(n-1), n, 6)$. Conversely, a strongly regular graph with these same parameters $(n^{2}, 3(n-1), n, 6)$ is called a \textit{pseudo-Latin square graph}. Bruck showed that for $n > 23$, a pseudo-Latin square graph is a Latin square graph \cite{Bruck}.

Albert showed that a quasigroup $Q$ is isotopic to a group $G$ if and only if $Q$ is isomorphic to $G$. In general, isotopic quasigroups need not be isomorphic \cite{Albert}. \\

\noindent \textbf{Nets.} A \textit{net} or \textit{partial geometry} of order $n \geq 1$ and degree $k \geq 1$, denoted $\mathcal{N}(n, k)$, consists of a set $P$ of $n^{2}$ points and a set $L$ of $kn$ lines satisfying the following.
\begin{enumerate}[label=(\alph*)]
\item Each line in $L$ has $n$ points.
\item Each point in $P$ lies on exactly $k$ distinct lines.
\item The $kn$ lines of $L$ fall into $k$ parallel classes. No two lines in the same parallel class intersect. If $\ell_{1}, \ell_{2} \in L$ are in different parallel classes, then $\ell_{1}$ and $\ell_{2}$ share exactly one common point.
\end{enumerate}

Necessarily, $k \leq n+1$. If $k = n+1$, then the net is called an \textit{affine plane} of order $n$. A \textit{projective plane} is the extension of an affine plane $\mathcal{N}$ by adding $n+1$ new points $p_{1}, \ldots, p_{n+1}$ (one per parallel class). The point $p_{i}$ is then added to each line of the $i$th parallel class. We finally add a new line containing $\{ p_{1}, \ldots, p_{n+1}\}$. We may recover an affine plane from a projective plane by removing a line and the associated points from the projective plane. We assume that nets are given by the point-block incidence matrix.

Given a net $\mathcal{N}(n, k)$ where $k \leq n$, we may construct a \textit{net graph} of order $n$ and degree $k$ $G(V, E)$ on $n^{2}$ vertices corresponding to the points in $\mathcal{N}$. Two vertices are adjacent in $G$ if their corresponding points in $\mathcal{N}$ determine a line. We note that a net graph uniquely determines the net $\mathcal{N}(n, k)$ \cite{Bruck}. Miller showed that, provided $n > (k-1)^{2}$, this equivalence holds under polynomial-time reductions \cite{MillerTarjan}. A net graph is strongly-regular with parameters $(n^{2}, k(n-1), n-2+(k-1)(k-2), k(k-1))$. Conversely, a strongly regular graph with parameters $(n^{2}, k(n-1), n-2+(k-1)(k-2), k(k-1))$ is referred to as a \textit{pseudo-net graph}. Bruck showed that for $n$ sufficiently large, a pseudo-net graph of order $n$ and degree $k$ is a net graph \cite{Bruck}.

\noindent \\ \textbf{Designs}. Let $t \leq k \leq v$ and $\lambda$ be positive integers. A $(t, k, \lambda, v)$ design is an incidence structure $\mathcal{D} = (X, \mathcal{B}, I)$, where $X$ denotes our set of $v$ points. Now $\mathcal{B}$ is a subset of $\binom{X}{k}$, where each element of $\mathcal{B}$ is referred to as a \textit{block}, satisfying the property that each $t$-element subset of $X$ belongs to exactly $\lambda$ blocks. If $t < k < v$, we say that the design is \textit{non-trivial}. If $\lambda = 1$, the design is referred to as a \textit{Steiner design}. We denote Steiner designs as $(t, k, v)$-designs when we want to specify $v$ the number of points, or Steiner $(t, k)$-designs when referring to a family of designs. We note that Steiner $(2, 3)$-designs are known as \textit{Steiner triple systems}. Projective planes are Steiner $(2, q+1, q^{2} + q + 1)$-designs, and affine planes are Steiner $(2, q, q^{2})$-designs. We assume that designs are given by the point-block incidence matrix. 

Let $\mathcal{D} = (X, \mathcal{B}, I)$ be a design, and let $A \subseteq X$. The \textit{derived design} at $A$, denoted $\mathcal{D}(A)$, has the set of points $X \setminus A$ and blocks $\{ B \setminus A : B \in \mathcal{B}, A \subsetneq B\}$. If $\mathcal{D}$ is a Steiner $(t, k, v)$-design, then $\mathcal{D}(A)$ is a Steiner $(t-|A|, k-|A|, v-|A|)$-design. 

For a design $\mathcal{D} = (X, \mathcal{B}, I)$, we may define a \textit{block intersection graph} (also known as a \textit{line graph}) $G(V, E)$, where $V(G) = \mathcal{B}$ and two blocks $B_{1}, B_{2}$ are adjacent in $G$ if and only if $B_{1} \cap B_{2} \neq \emptyset$. In the case of a Steiner $2$-design, the block-intersection graph is strongly regular. For Steiner triple systems, the block-intersection graphs are strongly regular with parameters $(n(n-1)/6, 3(n-3)/2, (n+3)/2, 9)$. Conversely, strongly regular graphs with the parameters $(n(n-1)/6, 3(n-3)/2, (n+3)/2, 9)$ are referred to as \textit{pseudo-STS graphs}. Bose showed that pseudo-STS graphs graphs with strictly more than $67$ vertices are in fact STS graphs \cite{Bose}.

\subsection{Computational Complexity} \label{sec:Complexity}

We assume familiarity with Turing machines and the complexity classes $\textsf{P}, \textsf{NP}, \textsf{L}$, and $\textsf{NL}$- we refer the reader to standard references \cite{ComplexityZoo, AroraBarak}. \\

\noindent \textbf{Circuit Complexity.} For a standard reference, see \cite{VollmerText}. We consider Boolean circuits using the \textsf{AND}, \textsf{OR}, \textsf{NOT}, and \textsf{Majority}, where $\textsf{Majority}(x_{1}, \ldots, x_{n}) = 1$ if and only if $\geq n/2$ of the inputs are $1$. Otherwise, $\textsf{Majority}(x_{1}, \ldots, x_{n}) = 0$. In this paper, we will consider \textit{logspace uniform} circuit families $(C_{n})_{n \in \mathbb{N}}$, in which a deterministic logspace Turing machine can compute the map $1^{n} \mapsto \langle C_{n} \rangle$ (here, $\langle C_{n} \rangle$ denotes an encoding of the circuit $C_{n}$).

\begin{definition}
Fix $k \geq 0$. We say that a language $L$ belongs to (logspace uniform) $\textsf{NC}^{k}$ if there exist a (logspace uniform) family of circuits $(C_{n})_{n \in \mathbb{N}}$ over the $\textsf{AND}, \textsf{OR}, \textsf{NOT}$ gates such that the following hold:
\begin{itemize}
\item The $\textsf{AND}$ and $\textsf{OR}$ gates take exactly $2$ inputs. That is, they have fan-in $2$.
\item $C_{n}$ has depth $O(\log^{k} n)$ and uses (has size) $n^{O(1)}$ gates. Here, the implicit constants in the circuit depth and size depend only on $L$.

\item $x \in L$ if and only if $C_{|x|}(x) = 1$. 
\end{itemize}
\end{definition}

\noindent The complexity class $\textsf{AC}^{k}$ is defined analogously as $\textsf{NC}^{k}$, except that the $\textsf{AND}, \textsf{OR}$ gates are permitted to have unbounded fan-in. That is, a single $\textsf{AND}$ gate can compute an arbitrary conjunction, and a single $\textsf{OR}$ gate can compute an arbitrary disjunction. The complexity class $\textsf{TC}^{k}$ is defined analogously as $\textsf{AC}^{k}$, except that our circuits are now permitted $\textsf{Majority}$ gates of unbounded fan-in. For every $k$, the following containments are well-known:
\[
\textsf{NC}^{k} \subseteq \textsf{AC}^{k} \subseteq \textsf{TC}^{k} \subseteq \textsf{NC}^{k+1}.
\]

\noindent In the case of $k = 0$, we have that:
\[
\textsf{NC}^{0} \subsetneq \textsf{AC}^{0} \subsetneq \textsf{TC}^{0} \subseteq \textsf{NC}^{1} \subseteq \textsf{L} \subseteq \textsf{NL} \subseteq \textsf{AC}^{1}.
\]

\noindent We note that functions that are $\textsf{NC}^{0}$-computable can only depend on a bounded number of input bits. Thus, $\textsf{NC}^{0}$ is unable to compute the $\textsf{AND}$ function. It is a classical result that $\textsf{AC}^{0}$ is unable to compute \algprobm{Parity} \cite{Smolensky87algebraicmethods}. The containment $\textsf{TC}^{0} \subseteq \textsf{NC}^{1}$ (and hence, $\textsf{TC}^{k} \subseteq \textsf{NC}^{k+1}$) follows from the fact that $\textsf{NC}^{1}$ can simulate the \textsf{Majority} gate. The class $\textsf{NC}$ is:
\[
\textsf{NC} := \bigcup_{k \in \mathbb{N}} \textsf{NC}^{k} = \bigcup_{k \in \mathbb{N}} \textsf{AC}^{k} = \bigcup_{k \in \mathbb{N}} \textsf{TC}^{k}.
\]

\noindent It is known that $\textsf{NC} \subseteq \textsf{P}$, and it is believed that this containment is strict. 

Let $d : \mathbb{N} \to \mathbb{N}$ be a function. The complexity class $\textsf{FO}(d(n))$ is the set of languages decidable by uniform $\textsf{AC}$ circuit families of depth $O(d(n))$ and polynomial size. The complexity class $\textsf{FOLL} = \textsf{FO}(\log \log n)$. It is known that $\textsf{AC}^{0} \subsetneq \textsf{FOLL} \subsetneq \textsf{AC}^{1}$, and it is open as to whether $\textsf{FOLL}$ is contained in $\textsf{NL}$ \cite{BKLM}.

For a complexity class $\mathcal{C}$, we define $\beta_{i}\mathcal{C}$ to be the set of languages $L$ such that there exists an $L' \in \mathcal{C}$ such that $x \in L$ if and only if there exists $y$ of length at most $O(\log^{i} |x|)$ such that $(x, y) \in L'$. For any $i, c \geq 0$, $\beta_{i}\textsf{FO}((\log \log n)^{c})$ cannot compute \algprobm{Parity} \cite{ChattopadhyayToranWagner}. \\

\noindent \textbf{Reductions.} We will consider several notions of reducibility. Let $L_{1}, L_{2}$ be languages. We say that $L_{1}$ is \textit{many-one} reducible to $L_{2}$, denoted $L_{1} \leq_{m} L_{2}$, if there exists a computable function $\varphi : \{0,1\}^{*} \to \{0,1\}^{*}$ such that $x \in L_{1} \iff \varphi(x) \in L_{2}$. We will often ask that, for a given complexity $\mathcal{C}$, $\varphi$ is $\mathcal{C}$-computable. 

We now turn towards recalling the notion of a Turing reduction. Intuitively, we say that the language $L_{1}$ is \textit{Turing-reducible} to $L_{2}$ if, given an algorithm to solve $L_{2}$, we can design an algorithm to solve $L_{1}$. This is made precise with oracle Turing machines, which we recall now. An \textit{oracle Turing machine} is a Turing machine with a separate oracle tape and oracle query state. When the Turing machine enters the query state, it transitions to one of two specified states based on whether the string on the oracle tape belongs to the oracle. When the oracle is unspecified, we write $M^{\square}$. When the oracle $\mathcal{O}$ is specified, we write $M^{\mathcal{O}}$. We now make precise the notion of a Turing reduction. 

A \textit{Turing reduction} from $L_{1}$ to $L_{2}$ is an oracle Turing machine $M^{\square}$ such that $M^{L_{2}}$ decides $L_{1}$. If there exists a Turing reduction from $L_{1}$ to $L_{2}$, we write $L_{1} \leq_{T} L_{2}$. We will be particularly interested in Turing reductions from $L_{1}$ to $L_{2}$, where $M^{L_{2}}$ runs in polynomial time. 
 
We finally recall the notion of a truth-table reduction. Again let $L_{1}, L_{2}$ be languages. We say that $L_{1}$ is \textit{truth-table} reducible to $L_{2}$, denoted $L_{1} \leq_{tt} L_{2}$, if there exists a computable function $g$ mapping an input $X$ to inputs $Y_{1}, \ldots, Y_{k}$, and a computable function $f$ (called the \textit{truth-table predicate}), which depends on $X$ and maps $\{0,1\}^{k} \to \{0,1\}$ such that $X \in L_{1}$ if and only if $f(\chi_{L_{2}}(Y_{1}), \ldots, \chi_{L_{2}}(Y_{k})) = 1$ (where $\chi_{L_{2}}$ is the characteristic function for $L_{2}$). The truth-table reduction is \textit{disjunctive} if $f$ is the $\textsf{OR}$ function. We say that the truth-table reduction is $\mathcal{C}$-computable if both $f, g$ are $\mathcal{C}$-computable. 

It is known that:
\[
L_{1} \leq_{m} L_{2} \implies L_{1} \leq_{tt} L_{2} \implies L_{1} \leq_{T} L_{2},
\]

\noindent and that this chain of implications holds even in the presence of resource constraints (e.g., if we were to insist that all steps be $\textsf{P}$-computable). Some care has to be taken to formulate a precise notion of a Turing reduction for circuit classes.  Fortunately, we will not need the notion of a Turing reduction for circuit classes, and so we will not discuss this further.

\noindent \\ \textbf{Graph Isomorphism.} We refer to \algprobm{GI} as the decision problem \algprobm{Graph Isomorphism}, which takes as input two graphs and asks whether there is an isomorphism between them. The complexity class \textsf{GI} is the set of decision problems that are polynomial-time Turing reducible to \algprobm{GI}.

\subsection{Color Refinement}
Let $G$ be a graph. The \textit{Color Refinement} (or $1$-dimensional Weisfeiler--Leman) algorithm works by iteratively coloring the vertices of $G$ in an isomorphism invariant manner. The algorithm begins by assigning each vertex an initial color based on their degree; that is, $\chi_{0}(v) = \text{deg}(v)$. Now for $r \geq 0$, we refine the coloring in the following manner:
\[
\chi_{r+1}(u) = (\chi_{r}(u), \{\!\{ \chi_{r}(v) : v \in N(u) \}\!\}),
\]

\noindent where $\{\!\{ \cdot \}\!\}$ denotes a multiset. The \textit{count-free} variant of Color Refinement works analogously, except the refinement step considers the set rather than multiset of colors at each round: 
\[
\chi_{r+1}(u) = (\chi_{r}(u), \{ \chi_{r}(v) : v \in N(u) \}).
\]

\noindent The algorithm terminates when the partition on the vertices induced by the coloring is not refined. We can use Color Refinement can as a non-isomorphism test by running the algorithm on $G \dot \cup H$ and comparing the multiset of colors. Note that if at round $r$, if the multiset of colors for $G$ differs from $H$, then we can conclude that $G \not \cong H$. Each iteration of the standard counting Color Refinement can be implemented using a logspace uniform $\textsf{TC}^{0}$-circuit, and each iteration of the count-free Color Refinement can be implemented using a logspace uniform $\textsf{AC}^{0}$-circuit (c.f., the work of Grohe \& Verbitsky \cite{GroheVerbitsky} who observed this for $k$-dimensional Weisfeiler--Leman where $k \geq 2$. The analogous parallel implementation holds for Color Refinement).

The \textit{individualize-and-refine} paradigm works as follows. Let $S$ be a sequence of vertices in $G$. We first assign each vertex in $S$ a unique color (so no two vertices in $S$ have the same color, and no vertex in $V(G) - S$ has the same color as any vertex in $S$), and then run Weisfeiler--Leman. The refinement step is usually performed using Color Refinement ($1$-WL), but may be performed using higher-dimensional Weisfeiler--Leman. For the purposes of this paper, we will use the count-free variant of Color Refinement in the refinement step. 

It is known that even higher-dimensional Weisfeiler--Leman fails to place \algprobm{GI} into $\textsf{P}$ \cite{CFI, NeuenSchweitzerIR}. We refer to Sandra Kiefer's dissertation \cite{KieferThesis} for a thorough and current survey of Weisfeiler--Leman.

\section{Latin Square Isotopy}

In this section, we show that the \algprobm{Latin Square Isotopy} problem is in $\beta_{2}\textsf{L} \cap \beta_{2}\textsf{FOLL}$.   The key technique is to guess cube generating sequences $A$ and $B$ for the Latin square $L_{1}$, and cube generating sequences $A^{\prime}, B^{\prime}$ for the Latin square $L_{2}$. We then (attempt to) construct appropriate bijections $\alpha, \beta : L_{1} \to L_{2}$, where $\alpha$ extends the map from $A \to A^{\prime}$ and $\beta$ extends the map from $B \to B^{\prime}$. We can construct such bijections in $\textsf{FOLL} \cap \textsf{L}$ using the techniques from Chattopadhyay, Tor\'an, and Wagner \cite{ChattopadhyayToranWagner}. 

The key step remains in checking whether the map sending the product of each pair $(x, y) \in L_{1} \times L_{1}$, $xy \mapsto \alpha(x)\beta(y)$ a bijection. Wolf approaches this in the following manner. First, construct sets $C = \{ a_{i}b_{i} : a_{i} \in A, b_{i} \in B\}$ and $C^{\prime} = \{a_{i}^{\prime}b_{i}^{\prime} : a_{i}^{\prime} \in A^{\prime}, b_{i}^{\prime} \in B^{\prime}\}$. Then check whether the map $a_{i}b_{i} \mapsto a_{i}^{\prime}b_{i}^{\prime}$ extends to a bijection $\gamma : L_{1} \to L_{2}$. Finally, check whether $\alpha(x)\beta(y) = \gamma(xy)$ for all $x, y \in L_{1}$. We note that Wolf is able to do this in $\textsf{NC}^{2}$ \cite{Wolf}. If $C$ and $C^{\prime}$ are cube generating sequences, then we would be able to apply the technique of Chattophadyay, Tor\'an, and Wagner \cite{ChattopadhyayToranWagner} to compute the bijection $\gamma$ in $\textsf{FOLL} \cap \textsf{L}$. However, $C$ and $C^{\prime}$ need not be cube generating sequences. As an example in the case of groups, suppose that $B = A^{-1}$. Then $B$ is a cube generating sequence. Furthermore, $a_{i}b_{i} = 1$ for each $i$, and so $C$ has only the identity element. In general, we do not expect $C$ and $C^{\prime}$ to be cube generating sequences, even if they do generate $L_{1}$ and $L_{2}$ respectively. 

Instead of extending Wolf's technique, we extend Miller's technique in his second proof that\\ $\algprobm{Latin Square Isotopy}$ is decidable in time $n^{\log(n) + O(1)}$. Miller constructs the relation 
\[
R = \{ (xy, \alpha(x)\beta(y)) : x, y \in L_{1}\},
\]

\noindent and checks whether $R$ is a bijection. If $R$ is a bijection; then by construction, the triple $(\alpha, \beta, R)$ is an isotopy \cite[Theorem 2, Proof 2]{MillerTarjan}. Using the fact that $A$ and $B$ are cube generating sequences, we can compute $x$ and $y$ in $\textsf{L} \cap \textsf{FOLL}$. In the process of computing $\alpha$ and $\beta$, we obtain a data structure which allows us to compute $\alpha(x)$ and $\beta(y)$ in $\textsf{L} \cap \textsf{FOLL}$.

\begin{theorem} \label{LatinSquareIsotopy}
$\algprobm{Latin Square Isotopy}$ is in $\beta_{2}\textsf{L} \cap \beta_{2}\textsf{FOLL}$.
\end{theorem}

We begin with the following lemmas.

\begin{lemma} \label{lem:Surjectivity}
Let $A, B$ be finite sets of the same size, and let $R \subseteq A \times B$ be a relation. We can decide whether $R$ is a well-defined surjection (and as $|A| = |B|$, consequently whether $R$ is a well-defined bijection) in $\textsf{AC}^{0}$.
\end{lemma}

\begin{proof}
For each $b \in B$, define the relation:
\[
Y(b) = \bigvee_{a \in A} 1_{(a, b) \in R}.
\]

\noindent Observe that $Y(b)$ is $\textsf{AC}^{0}$-computable. Now $R$ is surjective if and only if $Y(b) = 1$ for all $b \in B$, which is defined by the following condition:
\[
\varphi := \bigwedge_{b \in B} Y(b).
\]

\noindent Observe that $\varphi$ is $\textsf{AC}^{0}$ computable. 
\end{proof}

\begin{lemma} \label{LatinSquareBijection}
Let $L_{1}$ and $L_{2}$ be Latin squares of order $n$, and let $k = 4\lceil \log(n) \rceil$. Suppose that $(g_{0}, g_{1}, \ldots, g_{k})$ and $(h_{0}, \ldots, h_{k})$ are cube generating sequences for $L_{1}$ and $L_{2}$ respectively, with balanced parenthesization $P$. Deciding whether the map $g_{i} \mapsto h_{i}$ for all $i \in \{0, \ldots, k\}$ extends to a bijection is in $\textsf{L} \cap \textsf{FOLL}$. 
\end{lemma}

\begin{proof}
For each $(g, h) \in L_{1} \times L_{2}$, define $X(g, h) = 1$ if and only if there exists $(\epsilon_{1}, \ldots, \epsilon_{k}) \in \{0,1\}^{k}$ such that:
\begin{align*}
&g := g_{0}g_{1}^{\epsilon_{1}} \cdots g_{k}^{\epsilon_{k}}, \\
&h := h_{0}h_{1}^{\epsilon_{1}} \cdots h_{k}^{\epsilon_{k}}.
\end{align*}

Chattophadyay, Tor\'an, and Wagner showed that the cube words for $g$ and $h$ can be computed in $\textsf{L} \cap \textsf{FOLL}$  \cite{ChattopadhyayToranWagner}, so $X(g, h)$ is computable in $\textsf{L} \cap \textsf{FOLL}$. We note that the $\textsf{FOLL}$ bound follows from the fact that $P$ is a balanced parenthesization. As the two quasigroups are the same size, the map on the quasigroups induced by $(g_{0}, g_{1}, \ldots, g_{k}) \mapsto (h_{0}, \ldots, h_{k})$ is a well-defined bijection iff the induced map is injective iff the induced map is surjective. So it now suffices to check whether the induced map is surjective. By \Lem{lem:Surjectivity}, we may check whether the induced map is surjective in $\textsf{AC}^{0}$.
\end{proof}

\begin{proof}[Proof of Theorem \ref{LatinSquareIsotopy}]
Let $k = 4\lceil \log_{2}(n) \rceil$. We use $4k^{2}$ non-deterministic bits to guess cube generating sequences $A, B \subseteq L$ and $A^{\prime}, B^{\prime} \subseteq L^{\prime}$, where $A = \{a_{0}, a_{1}, \ldots, a_{k} \}$, $B = \{ b_{0}, b_{1}, \ldots, b_{k}\}$, $A^{\prime} = \{ a_{0}^{\prime}, a_{1}^{\prime}, \ldots, a_{k}^{\prime}\}$, and $B^{\prime} = \{ b_{0}^{\prime}, b_{1}^{\prime}, \ldots, b_{k}^{\prime}\}$. We may then in $\textsf{L} \cap \textsf{FOLL}$ check the following.
\begin{itemize}
\item We first check that the map $a_{i} \mapsto a_{i}^{\prime}$ extends to a bijection of $\langle A \rangle$ and $\langle A^{\prime} \rangle$. In particular, we may check in $\textsf{L} \cap \textsf{FOLL}$ whether $L_{1} = \text{Cube}(A)$ and $L_{2} = \text{Cube}(A^{\prime})$ (see \cite[Theorem 5]{ChattopadhyayToranWagner}). The procedure in Lemma \ref{LatinSquareBijection} that decides whether the map $a_{i} \mapsto a_{i}^{\prime}$ for all $i \in \{0, \ldots, k\}$ extends to a bijection $L_{1} \to L_{2}$, also explicitly computes a bijection $\varphi_{A} : L_{1} \to L_{2}$ if one exists. 

\item We proceed analogously for the map $b_{i} \mapsto b_{i}^{\prime}$ for all $i \in \{0, \ldots, k\}$. In the case that $L_{1} = \text{Cube}(B)$ and $L_{2} = \text{Cube}(B^{\prime})$, let $\varphi_{B} : L_{1} \to L_{2}$ be the bijection computed by Lemma \ref{LatinSquareBijection}.
\end{itemize}

Suppose now that the bijections $\varphi_{A}, \varphi_{B} : L_{1} \to L_{2}$ have been constructed. We now attempt to construct a relation $\varphi_{C} : L_{1} \to L_{2}$ in such a way that $\varphi_{C}$ is a bijection if and only if $(\varphi_{A}, \varphi_{B}, \varphi_{C})$ is an isotopism. For each pair $(\ell, m) \in L_{1} \times L_{2}$, define $X(\ell, m) = 1$ if and only if we define $\varphi_{C}$ to map $\ell \mapsto m$.  For each $(\epsilon_{1}, \ldots, \epsilon_{k}), (\nu_{1}, \ldots, \nu_{k}) \in \{0, 1\}^{k}$, we do the following:
\begin{enumerate}[label=(\alph*)]
\item Compute:
\begin{align*}
&g := a_{0}a_{1}^{\epsilon_{1}} \cdots a_{k}^{\epsilon_{k}}, \\
&h := b_{0}b_{1}^{\nu_{1}} \cdots b_{k}^{\nu_{k}}.
\end{align*}

We note that computing $g$ and $h$ can be done in $\textsf{L} \cap \textsf{FOLL}$ (see Chattopadhyay, Tor\`an, and Wagner \cite{ChattopadhyayToranWagner}).

\item Compute $\ell := gh$, $\varphi_{A}(g)$, $\varphi_{B}(h)$, and $m := \varphi_{A}(g)\varphi_{B}(h)$. We set $\varphi_{C}(\ell) = m$, which we indicate by defining $X(\ell, m) = 1$. The computations at this stage are computable using an $\textsf{AC}^{0}$ circuit. 

\item It remains to check whether $\varphi_{C}$ is a bijection. By \Lem{lem:Surjectivity}, we may test whether $\varphi_{C}$ is a bijection in $\textsf{AC}^{0}$. 
\end{enumerate}

Now $\varphi_{C}$ was constructed so that $(\varphi_{A}, \varphi_{B}, \varphi_{C})$ satisfies the homotopy condition, so, as they are also bijective, they are an isotopy. Thus, checking whether $L_{1}$ and $L_{2}$ are isotopic is in $\beta_{2}\textsf{L} \cap \beta_{2}\textsf{FOLL}$.
\end{proof}

Now for any $i, c \geq 0$, we have that $\beta_{i}\textsf{FO}((\log \log n)^{c})$ cannot compute \algprobm{Parity} \cite{ChattopadhyayToranWagner}. As \algprobm{Latin Square Isotopy} belongs to $\beta_{2}\textsf{FOLL}$, we obtain the following corollary.

\begin{corollary}
For any $i, c \geq 0$, \algprobm{GI} is not $\beta_{i}\textsf{FO}((\log \log n)^{c})$-reducible to \algprobm{Latin Square Isotopy}.
\end{corollary}

Miller \cite{MillerTarjan} showed that isomorphism testing of Latin square graphs is polynomial-time reducible to the Latin square isotopy problem. Wolf \cite{Wolf} improved this bound, showing that isomorphism testing of Latin square graphs is $\textsf{NC}^{1}$-reducible to testing for Latin square isotopy. We recall Wolf's result below.

\begin{lemma}[{\cite[Lemma 4.11]{Wolf}}]
Let $G$ be a Latin square graph derived from a Latin square of size $n$. We can retrieve a Latin square from $G$ with a polynomial-sized $\textsf{NC}$ circuit with $O(\log n)$ depth.
\end{lemma}

\begin{remark}
The statement of Lemma 4.11 in Wolf actually claims a depth of $O(\log^{2} n)$. However, in the proof of Lemma 4.11, Wolf shows that only $O(\log n)$ depth is needed \cite{Wolf}.
\end{remark}

Now by Remark~\ref{rmk:MainClass}, we can place the Latin squares into their main isotopy class in $\textsf{AC}^{0}$. In light of this, Wolf's result, Theorem \ref{LatinSquareIsotopy},  and Bruck's result that for $n > 23$, a pseudo-Latin square graph is a Latin square graph \cite{Bruck}, we obtain the following corollary.

\begin{corollary}
Isomorphism of pseudo-Latin square graphs can be decided in $\beta_{2}\textsf{L}$. 
\end{corollary}

\begin{remark}
To the best of the author's knowledge, the largest pseudo-Latin square graph that is not a  Latin square graph is the Hall--Janko graph, where $n = 10$ (that is, the Hall--Janko graph has $n^{2} = 100$ vertices) \cite{Suzuki, Nakasora}. The author would be grateful for references to larger examples.
\end{remark}

We next show that the reduction from \cite[Lemma 4.11]{Wolf}, which effectively parallelizes the reduction found in \cite{MillerTarjan}, can be implemented in $\textsf{AC}^{0}$. Furthermore, we can place the recovered Latin square into its main isotopy class in $\textsf{AC}^{0}$ (see Remark~\ref{rmk:MainClass}). It follows that \algprobm{Latin Square Graph Isomorphism} is also in $\beta_{2}\textsf{FOLL}$. As $\beta_{2}\textsf{FOLL}$ cannot compute \algprobm{Parity}, we obtain that for any $i, c \geq 0$, \algprobm{GI} is not $\beta_{i}\textsf{FO}((\log \log n)^{c})$-reducible to \algprobm{Latin Square Graph Isomorphism}.

\begin{lemma} \label{RecoverLSGAC0}
Let $G$ be a Latin square graph obtained from a Latin square of order $n$. We can recover a Latin square from $G$ using an $\textsf{AC}^{0}$ circuit.
\end{lemma}

\begin{proof}
We carefully analyze the reduction from \cite{MillerTarjan, Wolf}. By construction, two vertices $u$ and $v$ in $G$ are adjacent precisely if $u$ and $v$ correspond to elements in the Latin square that are either in the same row, the same column, or have the same value. As a result, each row, each column, and the nodes corresponding to a fixed given value all induce cliques of size $n$. The algorithm effectively identifies these cliques and their relations to the other cliques in order to recover a Latin square.

Denote $L = (\ell_{ij})_{1 \leq i, j \leq n}$ to be the $n \times n$ matrix, which our algorithm will populate with values for the Latin square. We proceed as follows.

\begin{enumerate}
\item We begin by selecting two adjacent vertices $x_{1}$ and $x_{2}$. Without loss of generality, we may assume $x_{1}$ and $x_{2}$ belong to the same row of the Latin square. We next find the $n$ vertices $v_{1}, \ldots, v_{n}$ adjacent to both $x_{1}$ and $x_{2}$. Precisely, for a vertex $v$, we have the indicator
\[
X(v) = E(x_{1}, v) \land E(x_{2}, v),
\]

where $E(u, v) = 1$ precisely if $u$ and $v$ are adjacent. Exactly $n$ indicators will evaluate to $1$. All but two of these nodes form a clique of size $n$ with $x_{1}$ and $x_{2}$. As $v_{1}, \ldots, v_{n}$ are adjacent to $x_{1}$ and $x_{2}$, it suffices to check which $n-2$ vertices of $v_{1}, \ldots, v_{n}$ form a clique of size $(n-2)$. For a given set $S \in \binom{[n]}{n-2}$, it suffices to check:
\[
\bigwedge_{\substack{i, j \in S \\ i \neq j}} E(v_{i}, v_{j}),
\]

which is $\textsf{AC}^{0}$-computable. There are $\binom{n}{n-2} \in \Theta(n^{2})$ such sets to check. Thus, identifying the clique is $\textsf{AC}^{0}$-computable. 

In parallel, we label the vertices of the clique as $x_{3}, \ldots, x_{n}$. One node not adjacent to any of $x_{3}, \ldots, x_{n}$ is labeled $y_{2}$. 

\item We associate $\ell_{1j}$ with $x_{j}$. Precisely, we set (in parallel) $\ell_{1j} = j$ for each $j \in [n]$. This step is $\textsf{AC}^{0}$-computable.

\item We next find the $n$-clique associated with $x_{1}$ and $y_{2}$, and label the additional vertices as $y_{3}, \ldots, y_{n}$. By similar argument as in Step 1, this step is $\textsf{AC}^{0}$-computable. Here, we view $x_{1}, y_{2}, y_{3}, \ldots, y_{n}$ as the first column of the Latin square.

\item For each $3 \leq i \leq n$, there exists a $3 \leq j \leq n$ such that $x_{i}$ and $y_{j}$ are adjacent. In particular, as $x_{i}$ and $y_{j}$ are neither in the same row nor the same column, it must be the case that $x_{i}$ and $y_{j}$ correspond to elements in the Latin square with the same value. It follows that our choice of $j$ is in fact unique. We reorder $y_{3}, \ldots, y_{n}$ so that $x_{i}$ is adjacent to $y_{i}$. This step is $\textsf{AC}^{0}$-computable.

\item For $2 \leq j \leq n$, we associate $\ell_{j1}$ with $y_{j}$. Precisely, we set (in parallel), $\ell_{j1} = j$ for each $2 \leq j \leq n$. This step is $\textsf{AC}^{0}$-computable.

\item For each of the remaining $(n-1)^{2}$ nodes $z$, we do the following in parallel:
\begin{enumerate}
\item If $z$ is adjacent to $x_{1}$, then the edge $\{ x_{1}, z\}$ is a value edge (as $z$ is not in the same row or column as $x_{1}$). So there exist unique $i, j > 1$ such that $z$ is adjacent to $x_{j}$ and $y_{i}$. In this case, we set $\ell_{ij} = 1$. This case is $\textsf{AC}^{0}$-computable.

\item Suppose that $z$ is not adjacent to $x_{1}$. As each row, each column, and each value induce an $n$-clique, there exist unique $i, j, k \in [n]$ such that $z$ is adjacent to $x_{j}, y_{i}, x_{k}$, and $y_{k}$. Unless $i = j = k$, we set $\ell_{ij} = k$. If $i = j = k$, we do nothing at this step and defer to step 7. This case is $\textsf{AC}^{0}$-computable.
\end{enumerate}

\item We note that step 6b does not account for diagonal entries where $\ell_{ii} = i$. To this end, we do the following. For each $i \geq 2$, we (in parallel) set $\ell_{ii}$ to be the value that does not appear in row $i$. This step is $\textsf{AC}^{0}$-computable.
\end{enumerate}

As we have a finite number of steps, each of which are $\textsf{AC}^{0}$-computable, it follows that we may recover a Latin square from $G$ using an $\textsf{AC}^{0}$ circuit.
\end{proof}

\begin{proposition}
Isomorphism testing of pseudo-Latin square graphs is in $\beta_{2}\textsf{FOLL}$. In particular, for any $i, c \geq 0$, \algprobm{GI} is not $\beta_{i}\textsf{FO}((\log \log n)^{c})$-reducible to isomorphism testing of pseudo-Latin square graphs.
\end{proposition}

\begin{proof}
We may handle the cases when $n \leq 23$ in $\textsf{AC}^{0}$. So suppose $n \geq 23$, and let $G$ and $H$ be pseudo-Latin square graphs. As $n \geq 23$, we have by Bruck that $G$ and $H$ are Latin square graphs \cite{Bruck}. By Lemma \ref{RecoverLSGAC0}, we may in $\textsf{AC}^{0}$ recover canonical Latin squares $L_{G}$ and $L_{H}$ corresponding to $G$ and $H$. Now by \cite[Lemma~3]{MillerTarjan}, $G \cong H$ if and only if $L_{G}$ and $L_{H}$ are main class isotopic. By Remark~\ref{rmk:MainClass}, we may place $L_{G}$ (respectively, $L_{H}$) into a normal form corresponding to its main isotopy class in $\textsf{AC}^{0}$. By Theorem \ref{LatinSquareIsotopy}, we can test whether $L_{G}$ and $L_{H}$ are isotopic in $\beta_{2}\textsf{FOLL}$. Chattopadhyay, Tor\'an, and Wagner showed that $\beta_{2}\textsf{FOLL}$ cannot compute \algprobm{Parity} \cite{ChattopadhyayToranWagner}. The result now follows.
\end{proof}

\section{Isomorphism Testing of Steiner Designs}
In this section, we show that for any $i, c \geq 0$, $\algprobm{GI}$ is not $\beta_{i}\textsf{FO}((\log \log n)^{c})$-reducible to isomorphism testing of several families of Steiner designs.

\subsection{Nets}

In this section, we show that for any $i, c \geq 0$, $\algprobm{GI}$ is not $\beta_{i}\textsf{FO}((\log \log n)^{c})$-reducible to isomorphism testing of nets or the corresponding strongly regular graphs. We note that projective and affine planes are special cases of nets. 

\begin{theorem}
Deciding whether two $k$-nets are isomorphic is in $\beta_{2}\textsf{L} \cap \beta_{2}\textsf{FOLL}$.
\end{theorem}

\begin{proof}
For $k = 0, 1, 2$, the pair $(k, n)$ determines the net uniquely, and so deciding isomorphism is trivial in these cases. So assume that $n+1 \geq k \geq 3$. Let $\mathcal{N}_{1}(n, k), \mathcal{N}_{2}(n, k)$ be nets. We now start by guessing three non-parallel lines $\ell_{a}, \ell_{b}, \ell_{c} \in L(\mathcal{N}_{1})$ and three non-parallel lines $\ell_{a}^{\prime}, \ell_{b}^{\prime}, \ell_{c}^{\prime} \in L(\mathcal{N}_{2})$. By definition, no two lines in the same parallel class share any points in common, and two lines in different classes share exactly one point in common. As there are $kn$ lines in $\mathcal{N}_{i}$, we only require $O(\log(kn))$ bits to identify a given line. As $k \leq n+1$, in fact only $O(\log(n))$ bits are required. 

We may identify whether two lines are disjoint in $\textsf{AC}^{0}$. In particular, we may check in $\textsf{AC}^{0}$ whether $\ell_{a}, \ell_{b}$, and $\ell_{c}$ (respectively, $\ell_{a}^{\prime}, \ell_{b}^{\prime}$, and $\ell_{c}^{\prime}$) belong to different parallel classes. Suppose now that $\ell_{a}, \ell_{b}$, and $\ell_{c}$ (respectively, $\ell_{a}^{\prime}, \ell_{b}^{\prime}$, and $\ell_{c}^{\prime}$) belong to different parallel classes. As there are $\binom{kn}{2}$ such pairs to check in a given $\mathcal{N}_{i}$, we may identify in $\textsf{AC}^{0}$ the lines belonging to the three parallel classes $a, b, c$ in each $\mathcal{N}_{i}$.

Now our three parallel classes determine a net $X_{1}(n, 3)$ in $\mathcal{N}_{1}$ and a net $X_{2}(n, 3)$ in $\mathcal{N}_{2}$. We note that a net of degree $3$ identifies a quasigroup (Latin square) up to isomorphism. As \algprobm{Quasigroup Isomorphism} is in $\beta_{2}\textsf{L} \cap \beta_{2}\textsf{FOLL}$ \cite[Theorem 3.4]{ChattopadhyayToranWagner}, we may now in $\beta_{2}\textsf{L} \cap \beta_{2}\textsf{FOLL}$ decide whether $X_{1} \cong X_{2}$. We note that \cite[Theorem 3.4]{ChattopadhyayToranWagner} actually yields an isomorphism $\varphi : X_{1} \to X_{2}$. We may now in $\textsf{L} \cap \textsf{FOLL}$ check whether $\varphi$ extends to an isomorphism of $\mathcal{N}_{1}$ and $\mathcal{N}_{2}$. The result follows. 
\end{proof}

\begin{corollary}
For any $i, c \geq 0$, $\algprobm{GI}$ is not $\beta_{i}\textsf{FO}((\log \log n)^{c})$-reducible to isomorphism testing of two $k$-nets.
\end{corollary}

Miller previously showed that isomorphism testing of nets and the corresponding net graphs are equivalent under polynomial-time reductions when $n > (k-1)^{2}$ \cite[Theorem~8]{MillerTarjan}. A closer analysis shows that this equivalence holds under $\textsf{TC}^{0}$-reductions in general; and that when $k$ is bounded, the equivalence holds under $\textsf{AC}^{0}$-reductions.

\begin{lemma} \label{LemmaReconstructNets}
Suppose that $G(V, E)$ is a $k$-net graph of order $n$ and $n > (k-1)^{2}$. We can reconstruct the net $\mathcal{N}(n, k)$ associated with $G$ in $\textsf{TC}^{0}$. If $k$ is bounded, we reconstruct the net $\mathcal{N}(n, k)$ associated with $G$ in $\textsf{AC}^{0}$.
\end{lemma}

\begin{proof}
Suppose that $x_{1}, x_{2} \in V(G)$ are adjacent. As there are $kn$ lines, it suffices to show that in $\textsf{AC}^{0}$, we can identify the remaining vertices of $V(G)$ that are on the same line as $x_{1}, x_{2}$. We first note that we can, in $\textsf{AC}^{0}$, identify the vertices adjacent to both $x_{1}$ and $x_{2}$. 

Let $H$ be the subgraph induced by these vertices together with $x_{1}$ and $x_{2}$. We note that $H$ contains the maximum clique (of $n$ vertices) containing both $x_{1}$ and $x_{2}$. The vertices of this clique have degree $n-2$. As two adjacent vertices have $(n-2) + (k-1)(k-2)$ neighbors in common, there are $(k-1)(k-2)$ vertices of $H$ that do not belong to this clique. Recall that a net has $k$ parallel classes, and any two lines in a given parallel class have empty intersection. Furthermore, two lines from different parallel classes share exactly one point in common. Thus, each nonclique vertex of $H$ is adjacent to exactly $(k-1)$ elements of the clique. Thus, each nonclique vertex of $H$ has degree at most:
\[
(k-1) + (k-1)(k-2) - 1 = (k-1)^{2} - 1.
\]

In general, using the difference in degrees, we may identify the clique and nonclique vertices in $\textsf{TC}^{0}$. If $k$ is bounded, then we may identify the clique and nonclique vertices in $\textsf{AC}^{0}$. The result follows.
\end{proof}

We combine Lemma \ref{LemmaReconstructNets} with Bruck's result \cite[Theorem~7]{Bruck} that for fixed $k$, pseudo-net graphs with sufficiently many vertices are net graphs to obtain the following.

\begin{corollary}
For fixed $k$, isomorphism testing of pseudo-net graphs of order $n$ and degree $k$ is in $\beta_{2}\textsf{L} \cap \beta_{2}\textsf{FOLL}$. In particular, for any $i, c \geq 0$, \algprobm{GI} is not $\beta_{i}\textsf{FO}((\log \log n)^{c})$-reducible to isomorphism testing of pseudo-net graphs of order $n$ and degree $k$.
\end{corollary}

\begin{remark}
Let $p(x) := (1/2)x^{4} + x^{3} + x^{2} + (3/2)x$. Bruck \cite[Theorem~7]{Bruck} showed that if $n > p(k-1)$ and $k > 1$, then a pseudo-net graph of order $n$ and degree $k$ is a net graph.
\end{remark}

\subsection{Steiner Triple Systems}

Using the standard connection between Steiner triple systems and quasigroups in tandem with \cite{ChattopadhyayToranWagner}, we observe the following.

\begin{observation} \label{ThmSTSIsomorphism}
Deciding whether two Steiner triple systems are isomorphic is in $\beta_{2}\textsf{L} \cap \beta_{2}\textsf{FOLL}$.
\end{observation}

\begin{proof}
Let $\mathcal{S}$ be a Steiner triple system of order $n$. We define a quasigroup $Q$ on the set $[n]$, with the multiplication operation satisfying the following: (i) for each $x \in [n]$, define $x^{2} = x$, and (ii) for each block $\{a, b, c\}$, define $ab = c$ (note that as the blocks are unordered, all such products $ba = c, ac = ca = b, bc = cb = a$ are required). The Steiner triple system determines $Q$ up to isomorphism. In particular, this construction is $\textsf{AC}^{0}$-computable. As \algprobm{Quasigroup Isomorphism} belongs to $\beta_{2}\textsf{L} \cap \beta_{2}\textsf{FOLL}$ \cite{ChattopadhyayToranWagner}, it follows that deciding whether two Steiner triple systems are isomorphic is in $\beta_{2}\textsf{L} \cap \beta_{2}\textsf{FOLL}$.
\end{proof}

\begin{proposition} \label{ReconstructSteinerDesign}
Let $G$ be a block-incidence graph on $n$ vertices derived from a Steiner $2$-design, in which each block has $k$ points and $\sqrt{n} - 2 > (k-1)^{2}$. We can reconstruct the Steiner $2$-design in $\textsf{TC}^{0}$. Furthermore, if $k$ is bounded, then we may reconstruct the Steiner $2$-design in $\textsf{AC}^{0}$.
\end{proposition}

\begin{proof}
We closely analyze the proof of \cite[Proposition~10]{SpielmanSRGs}. We note that $n$ is the number of blocks in the Steiner design. Let $v$ be the number of points in the Steiner design, and let $R = (v-1)/(k-1)$ be the number of blocks containing a given point. Let $B_{1}, B_{2}$ be two blocks that intersect uniquely at the point $p$. There are $R-2+(k-1)^{2}$ blocks that intersect both $B_{1}$ and $B_{2}$. We note that $R-2$ of these blocks also go through $p$, and the remaining $(k-1)^{2}$ blocks go through points other than $p$. 

We note that $p$ corresponds to the edge $\{B_{1}, B_{2}\} \in E(G)$. The $R$ blocks that intersect $p$ (including $B_{1}, B_{2}$) form a clique. Furthermore, these $R-2$ blocks intersecting $B_{1}, B_{2}$ at $p$ do not intersect with the remaining $(k-1)^{2}$ blocks that intersect with $B_{1}, B_{2}$ in points other than $p$. Let $N$ be the set of mutual neighbors for $B_{1}, B_{2}$. Now if $R-2 > (k-1)^{2}$, the $R$-clique is distinguished from the remaining $(k-1)^{2}$ blocks that don't contain $p$ by their degrees in the induced subgraph $G[N]$. We may distinguish these vertices in $\textsf{TC}^{0}$ in general. If $k$ is bounded, then we may distinguish these vertices in $\textsf{AC}^{0}$.

The fact that $R > \sqrt{n}$ was established in \cite[Proposition~10]{SpielmanSRGs}.
\end{proof}

We combine Proposition \ref{ReconstructSteinerDesign} with Bose's result that pseudo-STS graphs with strictly more than $67$ vertices are STS graphs \cite{Bose} to obtain the following.

\begin{corollary}
Deciding whether two pseudo-STS graphs are isomorphic is in $\beta_{2}\textsf{L} \cap \beta_{2}\textsf{FOLL}$. In particular, for any $i, c \geq 0$, $\algprobm{GI}$ is not $\beta_{i}\textsf{FO}((\log \log n)^{c})$-reducible to isomorphism testing of pseudo-STS graphs.
\end{corollary}

\begin{remark}
To the best of the author's knowledge, it remains open whether the bound of $67$ due to Bose \cite{Bose} can be decreased.
\end{remark}

\subsection{Reduction to Steiner $2$-Designs}

Babai \& Wilmes \cite{BabaiWilmes} previously exhibited a reduction from isomorphism testing of Steiner $t$-designs to the case of Steiner $2$-designs. A careful analysis shows that their reduction is $\beta_{2}\textsf{AC}^{0}$-computable. As a corollary, we obtain that $\algprobm{GI}$ is not $\textsf{AC}^{0}$-reducible to isomorphism testing of Steiner $(t, t+1)$-designs. 

\begin{observation}[c.f. {\cite{BabaiWilmes}}]
If isomorphism testing of Steiner $2$-designs belongs to $\beta_{2}\textsf{L} \cap \beta_{2}\textsf{FOLL}$, then testing isomorphism of Steiner $t$-designs also belongs to $\beta_{2}\textsf{L} \cap \beta_{2}\textsf{FOLL}$.
\end{observation}

\begin{proof}
Let $\mathcal{D}_{1} = (X_{1}, \mathcal{B}_{1}, I_{1})$ and $\mathcal{D}_{2} = (X_{2}, \mathcal{B}_{2}, I_{2})$ be Steiner $(t, k, v)$-designs. We begin by guessing subsequences $a_{1}, \ldots, a_{t-2} \in X_{1}$ and $b_{1}, \ldots, b_{t-2} \in X_{2}$. While we think of $a_{1}, \ldots, a_{t-2}$ and $b_{1}, \ldots, b_{t-2}$ as determining subsets $A = \{ a_{1}, \ldots, a_{t-2} \} \subseteq X_{1}$ and $B = \{ b_{1}, \ldots, b_{t-2} \} \subseteq X_{2}$, we stress that we guess both the elements \textit{and} an ordering. As $t \in O(\log n)$ (see the proof of \cite[Theorem~30]{BabaiWilmes}, in which Babai \& Wilmes cite \cite{ChaudhuriWilson}), guessing $A$ and $B$ requires $O(\log^{2} n)$ bits.  We may write down the derived designs $\mathcal{D}_{1}(A)$ and $\mathcal{D}_{2}(B)$ in $\textsf{AC}^{0}$ (see Section~\ref{sec:Preliminaries} for the definition of a derived design).

Suppose first that $\mathcal{D}_{1} \cong \mathcal{D}_{2}$. Let $\varphi : X_{1} \to X_{2}$ be an isomorphism. Then for any $A \subseteq X_{1}$ of size $t-2$, the derived designs $\mathcal{D}_{1}(A)$ and $\mathcal{D}_{2}(\varphi(A))$ are isomorphic.

Conversely, let $\psi : X_{1} \setminus A \to X_{2} \setminus B$ be an isomorphism of the derived designs $\mathcal{D}_{1}(A) \cong \mathcal{D}_{2}(B)$. Define:
\[
\widehat{\psi}(x) = \begin{cases} \psi(x) & : x \not \in A, \\ b_{i} & : x = a_{i} \in A. \end{cases} 
\]

We may easily check in $\textsf{AC}^{0}$ whether $\widehat{\psi}$ is an isomorphism of $\mathcal{D}_{1}$ and $\mathcal{D}_{2}$.  In particular, if isomorphism testing of Steiner $2$-designs belongs to $\beta_{2}\textsf{L} \cap \beta_{2}\textsf{FOLL}$, we may use the added non-determinism to guess an isomorphism of $\mathcal{D}_{1}(A) \cong \mathcal{D}_{1}(B)$ that lifts to an isomorphism of $\mathcal{D}_{1} \cong \mathcal{D}_{2}$. In total, we are still using at most $O(\log^{2} n)$ non-deterministic bits.
\end{proof}

We obtain the following corollaries.

\begin{corollary}
The problem of deciding whether two Steiner $(t, t+1)$-designs are isomorphic is $\beta_{2}\textsf{AC}^{0}$-reducible to the problem of finding an isomorphism of two Steiner triple systems.
\end{corollary}

Now deciding isomorphism testing of Steiner triple systems is $\textsf{AC}^{0}$-reducible to \algprobm{Quasigroup Isomorphism} (see Theorem \ref{ThmSTSIsomorphism}). Furthermore, Chattopadhyay, Tor\'an, and Wagner solved the search version of \algprobm{Quasigroup Isomorphism}; that is, their $\beta_{2}\textsf{L} \cap \beta_{2}\textsf{FOLL}$ procedure for \algprobm{Quasigroup Isomorphism} returns an isomorphism of the two quasigroups whenever an isomorphism exists. So in particular \algprobm{GI} is not $\textsf{AC}^{0}$-reducible to the search version of \algprobm{Quasigroup Isomorphism} \cite{ChattopadhyayToranWagner}. We may use the isomorphism for the quasigroups to construct an isomorphism of the Steiner triple systems, which may in turn be used to construct an isomorphism of the two Steiner $(t, t+1)$-designs. We summarize this observation with the following corollaries.

\begin{corollary}
The problem of deciding whether two Steiner $(t, t+1)$-designs are isomorphic is $\beta_{2}\textsf{AC}^{0}$-reducible to the problem of finding an isomorphism of two quasigroups.
\end{corollary}

\begin{corollary}
For any $i, c \geq 0$, \algprobm{GI} is not $\beta_{i}\textsf{FO}((\log \log n)^{c})$-reducible to the problem of deciding whether two Steiner $(t, t+1)$-designs are isomorphic.
\end{corollary}

\section{Conference Graphs} \label{sec:Conference}

In this section, we consider the complexity of identifying conference graphs. A \textit{conference} graph is a strongly regular graph with parameters $(n, (n-1)/2, (n-5)/4, (n-1)/4)$. For a graph $G$, a \textit{distinguishing set} $S$ is a subset of $V(G)$ such that for every $u, v \in V(G)$, we have that $u \in S, v \in S$, or $N(u) \cap S \neq N(v) \cap S$. We have the following observation.

\begin{observation} \label{obs:ColorRefinement}
Let $G$ be a graph, and let $S$ be a distinguishing set. After individualizing each vertex in $S$, we have that $2$ rounds of the count-free Color Refinement algorithm will assign each vertex in $G$ a unique color.
\end{observation}

Babai \cite[Lemma~3.2]{BabaiCanonicalLabeling1} (take $k = (n-5)/4$) showed that conference graphs admit distinguishing sets of size $O(\log n)$ (in fact, almost all such subsets of this size are distinguishing). As a consequence of this and Observation~\ref{obs:ColorRefinement}, we obtain the following.

\begin{theorem}
Let $G$ be a conference graph, and let $H$ be arbitrary. We can decide isomorphism between $G$ and $H$ in $\beta_{2}\textsf{AC}^{0}$.
\end{theorem}

\begin{proof}
We use $m := O(\log^{2} n)$ non-deterministic bits to guess a sequence $S = (s_{1}, \ldots, s_{m})$ for $G$-- while it would suffice for $\{ s_{1}, \ldots, s_{m}\}$ to be a distinguishing set for $G$, we only need that after individualizing the elements of $S$, two rounds of count-free Color Refinement will assign each vertex in $G$ a unique color. There may be non-distinguishing sets that acccomplish this. Let $S' = (s_{1}', \ldots, s_{m}')$ be the vertices of $H$ that were guessed. We now individualize $S$ and $S'$ so that $s_{i} \mapsto s_{i}'$ get the same color, and $s_{i}, s_{j}$ get different colors whenever $i \neq j$. The individualization step is $\textsf{AC}^{0}$-computable. We now run two rounds of count-free Color Refinement, which is $\textsf{AC}^{0}$-computable (c.f., \cite{GroheVerbitsky} and the discussion in Section~\ref{sec:Preliminaries}). Lastly, we check that each vertex in $G, H$ has a unique color, and whether the map induced by the colors is an isomorphism. This last step is $\textsf{AC}^{0}$-computable. The result now follows.
\end{proof}

In light of the previous work of Chattopadhyay, Tor\'an, \& Wagner  \cite{ChattopadhyayToranWagner}, we obtain the following corollary.

\begin{corollary}
For any $i, c \geq 0$, there is no $\beta_{i}\textsf{FO}((\log \log n)^{c})$-computable reduction from $\algprobm{GI}$ to isomorphism testing of conference graphs.
\end{corollary}

\section{Conclusion}

We showed that for any $i, c \geq 0$, \algprobm{GI} is not $\beta_{i}\textsf{FO}((\log \log n)^{c})$-reducible to several isomorphism problems characterized by the generator enumeration technique, including \algprobm{Latin Square Isotopy}, isomorphism testing of nets (which includes affine and projective planes), isomorphism testing of Steiner $(t, t+1)$-designs, and isomorphism testing of conference graphs. As a corollary, we obtained that \algprobm{GI} is not $\beta_{i}\textsf{FO}((\log \log n)^{c})$-reducible to isomorphism testing of Latin square graphs, $k$-net graphs (for fixed $k$), and the block-intersection graphs arising from Steiner triple systems. Our work leaves several directions for further research. 

In Proposition \ref{ReconstructSteinerDesign}, we showed that a Steiner $2$-design can be recovered from its block-incidence graph in $\textsf{AC}^{0}$ when the block size is bounded. Otherwise, the procedure is $\textsf{TC}^{0}$-computable, as we need to distinguish vertices by their degrees. As a step towards ruling out $\beta_{i}\textsf{FO}((\log \log n)^{c})$ reductions from $\algprobm{GI}$, it would be of interest to show that we can recover a Steiner $2$-design from its block-incidence graph in a complexity class that cannot compute \algprobm{Parity}. To this end, we ask the following.

\begin{question}
Can Steiner $2$-designs of unbounded block size be recovered from their block-incidence graphs in $\textsf{AC}^{0}$?
\end{question}

It would also be of interest to find additional families of Steiner $2$-designs where the isomorphism problem belongs to $\beta_{i}\textsf{FO}((\log \log n)^{c})$. As a starting point, we ask the following.

\begin{question}
For Steiner $2$-designs with bounded block size, can we decide isomorphism in $\beta_{2}\textsf{FOLL}$?
\end{question}

While a Steiner $2$-design $\mathcal{D}$ admits generating set $S$ of $O(\log v)$ points \cite{BabaiLuksCanonicalLabeling}, we have that $\Aut_{S}(\mathcal{D})$ may in general be non-trivial. This is the key barrier for the techniques here. 

Babai \& Wilmes \cite{BabaiWilmes} and Chen, Sun, \& Teng \cite{ChenSunTeng} independently showed that Steiner $2$-designs admit a set of $O(\log v)$ points where, once individualized, the Color Refinement algorithm assigns a unique color to each point. A priori, it seems plausible that only $\poly \log \log v$ iterations are required. However, Color Refinement takes into account the multiset of colors, and so each round can be implemented with a $\textsf{TC}^{0}$-circuit. In particular, Babai \& Wilmes rely crucially on counting  \cite[Target 2]{BabaiWilmes}. So we ask the following.

\begin{question}
Does there exist an absolute constant $c$, such that Steiner $2$-designs admit a set $S$ of $O(\log^{c} v)$ points where, after individuaizing $S$, the coloring from $\poly \log \log n$ rounds of count-free Color Refinement assigns each point a unique color?
\end{question}

In the remark following the proof of \cite[Lemma~3.2]{BabaiCanonicalLabeling1}, Babai outines a deterministic proof that leverages the greedy set cover algorithm (see \cite{LovaszGreedyCover}) to obtain a distinguishing set of the prescribed size. Now suppose that $G$ and $H$ are isomorphic, and the greedy set cover algorithm returns distinguishing sequences $S$ and $S'$ of the same size for $G, H$ respectively. A priori, $S$ and $S'$ need not be \textit{canonical} in the sense that there need not be an isomorphism $\varphi : V(G) \to V(H)$ mapping $\varphi(S) = S'$. Thus, we ask the following.

\begin{question} \label{q:5.4}
Is it possible to deterministically construct a canonical distinguishing set of the size prescribed by \cite[Lemma~3.2]{BabaiCanonicalLabeling1} for a graph in polynomial time? 
\end{question}

An answer of yes to Question~\ref{q:5.4} in tandem with the work in Section~\ref{sec:Conference} would immediately yield a polynomial-time isomorphism test for conference graphs. Babai's work \cite{BabaiCanonicalLabeling1} implies an $n^{O(\log n)}$-time algorithm, and to the best of my knowledge, no further improvements have been made to the runtime.

Alternatively, we ask the following.

\begin{question} \label{q:5.5}
Let $G$ be a conference graph. Does there exists a set of vertices $S$ of size $O(1)$ such that, after individualizing $S$ and running Color Refinement, each vertex of $G$ would receive a unique color?
\end{question}

An answer of yes to Question~\ref{q:5.5} would also yield a polynomial-time isomorphism test for conference graphs, using the individualize-and-refine paradigm. Even finding a such a set $S$ of size $o(\log n)$ would be a major advance, as it would yield an $n^{o(\log n)}$-time isomorphism test.

\bibliographystyle{alphaurl}
\bibliography{references}

\begin{thebibliography}{BKLM01}

\bibitem[AB09]{AroraBarak}
Sanjeev Arora and Boaz Barak.
\newblock {\em Computational Complexity: A Modern Approach}.
\newblock 01 2009.
\newblock \href {https://doi.org/10.1017/CBO9780511804090}
  {\path{doi:10.1017/CBO9780511804090}}.

\bibitem[AK06]{ArvindKurur}
V.~Arvind and Piyush~P. Kurur.
\newblock Graph isomorphism is in {SPP}.
\newblock {\em Information and Computation}, 204(5):835--852, 2006.
\newblock \href {https://doi.org/10.1016/j.ic.2006.02.002}
  {\path{doi:10.1016/j.ic.2006.02.002}}.

\bibitem[Alb43]{Albert}
A.~A. Albert.
\newblock Quasigroups. {I}.
\newblock {\em Transactions of the American Mathematical Society},
  54(3):507--519, 1943.
\newblock \href {https://doi.org/10.2307/1990259} {\path{doi:10.2307/1990259}}.

\bibitem[Bab80]{BabaiCanonicalLabeling1}
László Babai.
\newblock On the complexity of canonical labeling of strongly regular graphs.
\newblock {\em SIAM Journal on Computing}, 9(1):212--216, 1980.
\newblock \href {https://doi.org/10.1137/0209018} {\path{doi:10.1137/0209018}}.

\bibitem[Bab81]{BabaiCanonicalLabeling2}
Laszlo Babai.
\newblock On the order of uniprimitive permutation groups.
\newblock {\em Annals of Mathematics}, 113(3):553--568, 1981.
\newblock \href {https://doi.org/10.2307/2006997} {\path{doi:10.2307/2006997}}.

\bibitem[Bab14]{BabaiFunctorialSRGs}
L\'{a}szl\'{o} Babai.
\newblock On the automorphism groups of strongly regular graphs {I}.
\newblock In {\em Proceedings of the 5th Conference on Innovations in
  Theoretical Computer Science}, ITCS '14, page 359–368, New York, NY, USA,
  2014. Association for Computing Machinery.
\newblock \href {https://doi.org/10.1145/2554797.2554830}
  {\path{doi:10.1145/2554797.2554830}}.

\bibitem[Bab16]{BabaiGraphIso}
L\'{a}szl\'{o} Babai.
\newblock Graph isomorphism in quasipolynomial time [extended abstract].
\newblock In {\em S{TOC}'16---{P}roceedings of the 48th {A}nnual {ACM} {SIGACT}
  {S}ymposium on {T}heory of {C}omputing}, pages 684--697. ACM, New York, 2016.
\newblock Preprint of full version at \arXiv{1512.03547v2}{[cs.DS]}.
\newblock \href {https://doi.org/10.1145/2897518.2897542}
  {\path{doi:10.1145/2897518.2897542}}.

\bibitem[BE99]{BE99}
Hans~Ulrich Besche and Bettina Eick.
\newblock Construction of finite groups.
\newblock {\em J. Symb. Comput.}, 27(4):387--404, 1999.
\newblock \href {https://doi.org/10.1006/jsco.1998.0258}
  {\path{doi:10.1006/jsco.1998.0258}}.

\bibitem[BEO02]{BEO02}
Hans~Ulrich Besche, Bettina Eick, and E.A. O'Brien.
\newblock A millennium project: Constructing small groups.
\newblock {\em Intern. J. Alg. and Comput}, 12:623--644, 2002.
\newblock \href {https://doi.org/10.1142/S0218196702001115}
  {\path{doi:10.1142/S0218196702001115}}.

\bibitem[BH92]{BuhrmanHomer}
Harry Buhrman and Steven Homer.
\newblock Superpolynomial circuits, almost sparse oracles and the exponential
  hierarchy.
\newblock In R.~K. Shyamasundar, editor, {\em Foundations of Software
  Technology and Theoretical Computer Science, 12th Conference, New Delhi,
  India, December 18-20, 1992, Proceedings}, volume 652 of {\em Lecture Notes
  in Computer Science}, pages 116--127. Springer, 1992.
\newblock \href {https://doi.org/10.1007/3-540-56287-7\_99}
  {\path{doi:10.1007/3-540-56287-7\_99}}.

\bibitem[BKL83]{BabaiKantorLuksCFSG}
L.~Babai, W.~M. Kantor, and E.~M. Luks.
\newblock Computational complexity and the classification of finite simple
  groups.
\newblock In {\em 24th Annual Symposium on Foundations of Computer Science
  (sfcs 1983)}, pages 162--171, 1983.
\newblock \href {https://doi.org/10.1109/SFCS.1983.10}
  {\path{doi:10.1109/SFCS.1983.10}}.

\bibitem[BKLM01]{BKLM}
David A.~Mix Barrington, Peter Kadau, Klaus{-}J{\"{o}}rn Lange, and Pierre
  McKenzie.
\newblock On the complexity of some problems on groups input as multiplication
  tables.
\newblock {\em J. Comput. Syst. Sci.}, 63(2):186--200, 2001.
\newblock \href {https://doi.org/10.1006/jcss.2001.1764}
  {\path{doi:10.1006/jcss.2001.1764}}.

\bibitem[BL83]{BabaiLuksCanonicalLabeling}
L\'{a}szl\'{o} Babai and Eugene~M. Luks.
\newblock Canonical labeling of graphs.
\newblock In {\em Proceedings of the Fifteenth Annual ACM Symposium on Theory
  of Computing}, STOC '83, page 171–183, New York, NY, USA, 1983. Association
  for Computing Machinery.
\newblock \href {https://doi.org/10.1145/800061.808746}
  {\path{doi:10.1145/800061.808746}}.

\bibitem[Bos63]{Bose}
R.~C. Bose.
\newblock {Strongly regular graphs, partial geometries and partially balanced
  designs.}
\newblock {\em Pacific Journal of Mathematics}, 13(2):389 -- 419, 1963.
\newblock \href {https://doi.org/pjm/1103035734} {\path{doi:pjm/1103035734}}.

\bibitem[Bru63]{Bruck}
R.~H. Bruck.
\newblock {Finite nets. II. Uniqueness and imbedding.}
\newblock {\em Pacific Journal of Mathematics}, 13(2):421 -- 457, 1963.
\newblock \href {https://doi.org/10.2140/pjm.1963.13.421}
  {\path{doi:10.2140/pjm.1963.13.421}}.

\bibitem[BW13]{BabaiWilmes}
L\'aszl\'o Babai and John Wilmes.
\newblock Quasipolynomial-time canonical form for {S}teiner designs.
\newblock In {\em STOC 2013}, pages 261--270, New York, NY, USA, 2013.
  Association for Computing Machinery.
\newblock \href {https://doi.org/10.1145/2488608.2488642}
  {\path{doi:10.1145/2488608.2488642}}.

\bibitem[CC81]{ColbournColbournGI}
{Marlene J.} Colbourn and {Charles J.} Colbourn.
\newblock Concerning the complexity of deciding isomorphism of block designs.
\newblock {\em Discrete Applied Mathematics}, 3(3):155--162, July 1981.
\newblock \href {https://doi.org/10.1016/0166-218X(81)90012-3}
  {\path{doi:10.1016/0166-218X(81)90012-3}}.

\bibitem[CFI92]{CFI}
Jin-Yi Cai, Martin F\"{u}rer, and Neil Immerman.
\newblock An optimal lower bound on the number of variables for graph
  identification.
\newblock {\em Combinatorica}, 12(4):389--410, 1992.
\newblock Originally appeared in SFCS '89.
\newblock \href {https://doi.org/10.1007/BF01305232}
  {\path{doi:10.1007/BF01305232}}.

\bibitem[CG70]{CorneilGotlieb}
D.~G. Corneil and C.~C. Gotlieb.
\newblock An efficient algorithm for graph isomorphism.
\newblock {\em J. ACM}, 17(1):51–64, January 1970.
\newblock \href {https://doi.org/10.1145/321556.321562}
  {\path{doi:10.1145/321556.321562}}.

\bibitem[CH03]{CH03}
John~J. Cannon and Derek~F. Holt.
\newblock Automorphism group computation and isomorphism testing in finite
  groups.
\newblock {\em J. Symb. Comput.}, 35:241--267, March 2003.
\newblock \href {https://doi.org/10.1016/S0747-7171(02)00133-5}
  {\path{doi:10.1016/S0747-7171(02)00133-5}}.

\bibitem[CM80]{ColbournMathon}
Marlene~J. Colbourn and Rudolf~A. Mathon.
\newblock On cyclic {S}teiner 2-designs.
\newblock In C.C. Lindner and A.~Rosa, editors, {\em Topics on Steiner
  Systems}, volume~7 of {\em Annals of Discrete Mathematics}, pages 215--253.
  Elsevier, 1980.
\newblock \href {https://doi.org/10.1016/S0167-5060(08)70182-1}
  {\path{doi:10.1016/S0167-5060(08)70182-1}}.

\bibitem[Col79]{colborun_1979}
M~J Colbourn.
\newblock An analysis technique for \text{Steiner} triple systems.
\newblock {\em Proc. Tenth Southeastern Conf. Combin., Graph Theory,
  Computing}, page 289–303, 1979.

\bibitem[CST13]{ChenSunTeng}
Xi~Chen, Xiaorui Sun, and Shang-Hua Teng.
\newblock Multi-stage design for quasipolynomial-time isomorphism testing of
  steiner 2-systems.
\newblock In {\em Proceedings of the Forty-Fifth Annual ACM Symposium on Theory
  of Computing}, STOC '13, page 271–280, New York, NY, USA, 2013. Association
  for Computing Machinery.
\newblock \href {https://doi.org/10.1145/2488608.2488643}
  {\path{doi:10.1145/2488608.2488643}}.

\bibitem[CTW13]{ChattopadhyayToranWagner}
Arkadev Chattopadhyay, Jacobo Tor\'{a}n, and Fabian Wagner.
\newblock Graph isomorphism is not {$\rm AC^0$}-reducible to group isomorphism.
\newblock {\em ACM Trans. Comput. Theory}, 5(4):Art. 13, 13, 2013.
\newblock Preliminary version appeared in FSTTCS '10; ECCC Tech. Report
  TR10-117.
\newblock \href {https://doi.org/10.1145/2540088} {\path{doi:10.1145/2540088}}.

\bibitem[ELGO02]{ELGO02}
Bettina Eick, C.~R. Leedham-Green, and E.~A. O'Brien.
\newblock Constructing automorphism groups of {$p$}-groups.
\newblock {\em Comm. Algebra}, 30(5):2271--2295, 2002.
\newblock \href {https://doi.org/10.1081/AGB-120003468}
  {\path{doi:10.1081/AGB-120003468}}.

\bibitem[GV06]{GroheVerbitsky}
Martin Grohe and Oleg Verbitsky.
\newblock Testing graph isomorphism in parallel by playing a game.
\newblock In Michele Bugliesi, Bart Preneel, Vladimiro Sassone, and Ingo
  Wegener, editors, {\em Automata, Languages and Programming, 33rd
  International Colloquium, {ICALP} 2006, Venice, Italy, July 10-14, 2006,
  Proceedings, Part {I}}, volume 4051 of {\em Lecture Notes in Computer
  Science}, pages 3--14. Springer, 2006.
\newblock \href {https://doi.org/10.1007/11786986_2}
  {\path{doi:10.1007/11786986_2}}.

\bibitem[Hub11]{Huber}
Michael Huber.
\newblock Computational complexity of reconstruction and isomorphism testing
  for designs and line graphs.
\newblock {\em Journal of Combinatorial Theory, Series A}, 118(2):341--349,
  2011.
\newblock \href {https://doi.org/10.1016/j.jcta.2010.06.006}
  {\path{doi:10.1016/j.jcta.2010.06.006}}.

\bibitem[IPZ01]{ETH}
Russell Impagliazzo, Ramamohan Paturi, and Francis Zane.
\newblock Which problems have strongly exponential complexity?
\newblock {\em Journal of Computer and System Sciences}, 63(4):512--530, 2001.
\newblock \href {https://doi.org/10.1006/jcss.2001.1774}
  {\path{doi:10.1006/jcss.2001.1774}}.

\bibitem[Kie19]{KieferThesis}
Sandra Kiefer.
\newblock {\em Power and Limits of the {Weisfeiler--Leman} Algorithm}.
\newblock PhD thesis, RWTH Aachen University, 2019.
\newblock URL:
  \url{https://publications.rwth-aachen.de/record/785831/files/785831.pdf}.

\bibitem[KST92]{GILowPP}
Johannes K{\"{o}}bler, Uwe Sch{\"{o}}ning, and Jacobo Tor{\'{a}}n.
\newblock Graph isomorphism is low for {PP}.
\newblock {\em Comput. Complex.}, 2:301--330, 1992.
\newblock \href {https://doi.org/10.1007/BF01200427}
  {\path{doi:10.1007/BF01200427}}.

\bibitem[Lad75]{Ladner}
Richard~E. Ladner.
\newblock On the structure of polynomial time reducibility.
\newblock {\em J. ACM}, 22(1):155–171, January 1975.
\newblock \href {https://doi.org/10.1145/321864.321877}
  {\path{doi:10.1145/321864.321877}}.

\bibitem[LGR16]{GR16}
Fran{\c{c}}ois Le~Gall and David~J. Rosenbaum.
\newblock On the group and color isomorphism problems.
\newblock \arXiv{1609.08253}{[cs.CC]}, 2016.

\bibitem[Lov75]{LovaszGreedyCover}
L.~Lovász.
\newblock On the ratio of optimal integral and fractional covers.
\newblock {\em Discrete Mathematics}, 13(4):383--390, 1975.
\newblock \href {https://doi.org/10.1016/0012-365X(75)90058-8}
  {\path{doi:10.1016/0012-365X(75)90058-8}}.

\bibitem[Luk82]{LuksBoundedValence}
Eugene~M. Luks.
\newblock Isomorphism of graphs of bounded valence can be tested in polynomial
  time.
\newblock {\em Journal of Computer and System Sciences}, 25(1):42--65, 1982.
\newblock \href {https://doi.org/10.1016/0022-0000(82)90009-5}
  {\path{doi:10.1016/0022-0000(82)90009-5}}.

\bibitem[Mat79]{MATHON1979131}
Rudolf Mathon.
\newblock A note on the graph isomorphism counting problem.
\newblock {\em Information Processing Letters}, 8(3):131--136, 1979.
\newblock \href {https://doi.org/10.1016/0020-0190(79)90004-8}
  {\path{doi:10.1016/0020-0190(79)90004-8}}.

\bibitem[Mil78]{MillerTarjan}
Gary~L. Miller.
\newblock On the {$n^{\log n}$} isomorphism technique (a preliminary report).
\newblock In {\em Proceedings of the Tenth Annual ACM Symposium on Theory of
  Computing}, STOC '78, pages 51--58, New York, NY, USA, 1978. Association for
  Computing Machinery.
\newblock \href {https://doi.org/10.1145/800133.804331}
  {\path{doi:10.1145/800133.804331}}.

\bibitem[Nak06]{Nakasora}
Hiroyuki Nakasora.
\newblock Mutually orthogonal latin squares and self-complementary designs.
\newblock {\em Mathematical Journal of Okayama University}, 01 2006.

\bibitem[Neu79]{Neumaier}
A.~Neumaier.
\newblock Strongly regular graphs with smallest eigenvalue $-m$.
\newblock {\em Archiv der Mathematik}, 33:392--400, 1979.
\newblock \href {https://doi.org/10.1007/BF01222774}
  {\path{doi:10.1007/BF01222774}}.

\bibitem[NS18]{NeuenSchweitzerIR}
Daniel Neuen and Pascal Schweitzer.
\newblock An exponential lower bound for individualization-refinement
  algorithms for graph isomorphism.
\newblock In Ilias Diakonikolas, David Kempe, and Monika Henzinger, editors,
  {\em Proceedings of the 50th Annual {ACM} {SIGACT} Symposium on Theory of
  Computing, {STOC} 2018, Los Angeles, CA, USA, June 25-29, 2018}, pages
  138--150. {ACM}, 2018.
\newblock \href {https://doi.org/10.1145/3188745.3188900}
  {\path{doi:10.1145/3188745.3188900}}.

\bibitem[RCW75]{ChaudhuriWilson}
Dijen~K. Ray-Chaudhuri and Richard~M. Wilson.
\newblock {On $t$-designs}.
\newblock {\em Osaka Journal of Mathematics}, 12(3):737--744, 1975.
\newblock \href {https://doi.org/10.18910/7296} {\path{doi:10.18910/7296}}.

\bibitem[Ros13]{Rosenbaum2013BidirectionalCD}
David~J. Rosenbaum.
\newblock Bidirectional collision detection and faster deterministic
  isomorphism testing.
\newblock \arXiv{1304.3935}{[cs.DS]}, 2013.

\bibitem[Sch88]{Schoning}
Uwe Sch{\"{o}}ning.
\newblock Graph isomorphism is in the low hierarchy.
\newblock {\em Journal of Computer and System Sciences}, 37(3):312 -- 323,
  1988.
\newblock \href {https://doi.org/10.1016/0022-0000(88)90010-4}
  {\path{doi:10.1016/0022-0000(88)90010-4}}.

\bibitem[SD76]{SchmidtDruffel}
Douglas~C. Schmidt and Larry~E. Druffel.
\newblock A fast backtracking algorithm to test directed graphs for isomorphism
  using distance matrices.
\newblock {\em J. ACM}, 23(3):433–445, July 1976.
\newblock \href {https://doi.org/10.1145/321958.321963}
  {\path{doi:10.1145/321958.321963}}.

\bibitem[SF21]{SchrockFrongillo}
Tyler Schrock and Rafael Frongillo.
\newblock Computational complexity of k-block conjugacy.
\newblock {\em Theoretical Computer Science}, 856:21--40, 2021.
\newblock \href {https://doi.org/10.1016/j.tcs.2020.12.009}
  {\path{doi:10.1016/j.tcs.2020.12.009}}.

\bibitem[Smo87]{Smolensky87algebraicmethods}
Roman Smolensky.
\newblock Algebraic methods in the theory of lower bounds for boolean circuit
  complexity.
\newblock In Alfred~V. Aho, editor, {\em Proceedings of the 19th Annual {ACM}
  Symposium on Theory of Computing, 1987, New York, New York, {USA}}, pages
  77--82. {ACM}, 1987.
\newblock \href {https://doi.org/10.1145/28395.28404}
  {\path{doi:10.1145/28395.28404}}.

\bibitem[Spi96]{SpielmanSRGs}
Daniel~A. Spielman.
\newblock {Faster Isomorphism Testing of Strongly Regular Graphs}.
\newblock In {\em Proceedings of the Twenty-Eighth Annual ACM Symposium on
  Theory of Computing}, STOC '96, page 576–584, New York, NY, USA, 1996.
  Association for Computing Machinery.
\newblock \href {https://doi.org/10.1145/237814.238006}
  {\path{doi:10.1145/237814.238006}}.

\bibitem[Suz69]{Suzuki}
Michio Suzuki.
\newblock A simple group of order $448,345,497,600$, 1969.

\bibitem[Tan13]{TangThesis}
Bangsheng Tang.
\newblock {\em Towards Understanding Satisfiability, Group Isomorphism and
  Their Connections}.
\newblock PhD thesis, Tsinghua University, 2013.
\newblock URL:
  \url{http://papakonstantinou.org/periklis/pdfs/bangsheng_thesis.pdf}.

\bibitem[Tor04]{Toran}
Jacobo Tor{\'{a}}n.
\newblock On the hardness of graph isomorphism.
\newblock {\em {SIAM} J. Comput.}, 33(5):1093--1108, 2004.
\newblock \href {https://doi.org/10.1137/S009753970241096X}
  {\path{doi:10.1137/S009753970241096X}}.

\bibitem[Vol99]{VollmerText}
Heribert Vollmer.
\newblock {\em Introduction to Circuit Complexity - {A} Uniform Approach}.
\newblock Texts in Theoretical Computer Science. An {EATCS} Series. Springer,
  1999.
\newblock \href {https://doi.org/10.1007/978-3-662-03927-4}
  {\path{doi:10.1007/978-3-662-03927-4}}.

\bibitem[Wil19]{WilsonSubgroupProfiles}
James~B. Wilson.
\newblock The threshold for subgroup profiles to agree is logarithmic.
\newblock {\em Theory of Computing}, 15(19):1--25, 2019.
\newblock \href {https://doi.org/10.4086/toc.2019.v015a019}
  {\path{doi:10.4086/toc.2019.v015a019}}.

\bibitem[Wol94]{Wolf}
Marty~J. Wolf.
\newblock Nondeterministic circuits, space complexity and quasigroups.
\newblock {\em Theoretical Computer Science}, 125(2):295--313, 1994.
\newblock \href {https://doi.org/10.1016/0304-3975(92)00014-I}
  {\path{doi:10.1016/0304-3975(92)00014-I}}.

\bibitem[ZKT85]{Zemlyachenko1985}
V.~N. Zemlyachenko, N.~M. Korneenko, and R.~I. Tyshkevich.
\newblock Graph isomorphism problem.
\newblock {\em Journal of Soviet Mathematics}, 29:1426--1481, 1985.
\newblock \href {https://doi.org/10.1007/BF02104746}
  {\path{doi:10.1007/BF02104746}}.

\bibitem[Zoo]{ComplexityZoo}
Complexity zoo.
\newblock URL: \url{https://complexityzoo.net}.

\end{thebibliography}

\end{document}